\documentclass[smallextended]{svjour3}
\usepackage[margin=1in]{geometry}
\usepackage[dvipsnames]{xcolor}
\def\makeheadbox{\relax}

\usepackage{hyperref}
\usepackage{enumerate}

\newcommand*{\Prob}{\mathsf{Pr}}

\newcommand{\rev}[1]{#1} 
\newcommand{\revtwo}[1]{#1} %

\title{Classical symmetries and \revtwo{the Quantum Approximate Optimization Algorithm}}

\author{Ruslan Shaydulin        \and
        Stuart Hadfield         \and
        Tad Hogg                \and
        Ilya Safro %
}

\institute{R. Shaydulin \at
              Argonne National Laboratory, Lemont, IL 60439, USA \\
              Quantum Artificial Intelligence Lab (QuAIL), NASA Ames Research Center, Moffett Field, CA 94035, USA \\
              KBR, Houston, TX 77002, USA \\
              \email{ruslan@anl.gov}           %
          \and
          S. Hadfield \at
              Quantum Artificial Intelligence Lab (QuAIL), NASA Ames Research Center, Moffett Field, CA 94035, USA \\
              USRA Research Institute for Advanced Computer Science (RIACS), Mountain View, CA 94043, USA
          \and 
          T. Hogg \at
              Quantum Artificial Intelligence Lab (QuAIL), NASA Ames Research Center, Moffett Field, CA 94035, USA
          \and 
          I. Safro \at
              Computer and Information Sciences, University of Delaware, Newark, DE 19716, USA
}

\usepackage{natbib}
\usepackage{graphicx}
\usepackage{amsmath}
\usepackage{amssymb}
\usepackage{printlen} 
\makeatletter
\newcommand\thefontsize[1]{{#1 The current font size is: \f@size pt\par}}
\makeatother

\newtheorem{innercustomgeneric}{\customgenericname}
\providecommand{\customgenericname}{}
\newcommand{\newcustomtheorem}[2]{%
  \newenvironment{#1}[1]
  {%
   \renewcommand\customgenericname{#2}%
   \renewcommand\theinnercustomgeneric{##1}%
   \innercustomgeneric
  }
  {\endinnercustomgeneric}
}

\newcustomtheorem{customthm}{Theorem}
\newcustomtheorem{customcor}{Corollary}

\usepackage{braket}

\usepackage{todonotes}
\usepackage{changes}
\definechangesauthor[color=orange]{rs}
\definechangesauthor[color=red]{is}
\definechangesauthor[color=blue]{th}
\definechangesauthor[color=green]{sh}

\usepackage{subfig}
\usepackage{wrapfig}

\newcommand{\pmin}{p_{\min}}
\newcommand{\maxcut}{MaxCut}
\usepackage{longtable}
\usepackage{etoolbox}
\usepackage{lscape}
\makeatletter
\let\antilandscape\landscape

\def\LS@antirot{%
  \setbox\@outputbox\vbox{\hbox{\rotatebox{-90}{\box\@outputbox}}}}
\patchcmd{\antilandscape}{\LS@rot}{\LS@antirot}{}{}

\usepackage{siunitx}
\usepackage{afterpage}

\begin{document}

\maketitle

\begin{abstract}
We study the relationship between the Quantum Approximate Optimization Algorithm (QAOA) and the underlying symmetries of the objective function to be optimized. Our approach formalizes the connection between quantum symmetry properties of the QAOA dynamics and  the group of classical symmetries of the objective function. The connection is general and includes but is not limited to problems defined on graphs. We show a series of results exploring the connection and highlight examples of hard problem classes where a nontrivial symmetry subgroup can be obtained efficiently. In particular we show how classical objective function symmetries lead to invariant 
measurement outcome probabilities across states connected by such symmetries, independent of the choice of algorithm parameters or number of layers. 
To illustrate the power of the developed connection, we apply machine learning techniques towards predicting QAOA performance based on symmetry considerations. We provide numerical evidence that a small set of graph symmetry properties suffices to predict the minimum QAOA depth required to achieve a target approximation ratio on the MaxCut problem, in a practically important setting where QAOA parameter schedules are constrained to be linear and hence easier to optimize.
\end{abstract}

\section{Introduction}

Advances in quantum hardware such as the recent demonstration of
quantum supremacy%
~\cite{Arute2019} have %
paved the way for the %
deployment of %
algorithms %
that potentially offer %
quantum advantage in applications, that is, outperform state-of-the-art classical algorithms for some computationally challenging problem. Combinatorial optimization is a particularly attractive and important target domain as %
optimization problems are omnipresent in science and industry and the difficulty of such problems means one must often settle for approximation algorithms or heuristics in practice. The Quantum Approximate Optimization Algorithm (QAOA)~\cite{farhi2014quantum,hadfield2017quantum} is %
a leading candidate 
algorithmic paradigm 
leveraging near-term quantum hardware to approximately solve certain classes of hard combinatorial optimization problems. 
Published results to date have %
shown experimental 
demonstrations of 
QAOA %
on quantum hardware using up to 23 qubits~\cite{googleqaoaexperimental}. 
Nevertheless, although a number of recent results have begun to shed some light on the potential of QAOA as compared with classical algorithms and heuristics, the power of QAOA is not yet well understood.

Nontrivial performance guarantees for QAOA circuits have thus far been obtained only in very limited settings, %
with the majority of rigorous bounds to the approximation ratio achievable by QAOA %
obtained only for small (high-level QAOA circuit depth)~$p$, primarily $p=1$ or $p=2$~\cite{farhi2014quantum,farhi2014quantumbounded,wang2018quantum,hadfield2018quantum, Wurtz2020bounds}. A notable exception is~\cite{jiang2017near}, which shows that QAOA-related circuits can recover the well-known quadratic speedup for Grover's unstructured search problem. Furthermore, such performance results are typically worst case and do not take advantage of the structure of a given problem instance, %
nor do they give much insight into the performance of deeper QAOA circuits.
Indeed, initial numerical studies on  small problem instances have indicated that going beyond small $p$ appears necessary to obtain potential quantum advantage~\cite{Shaydulin2019EvaluatingDOI,wang2018quantum,zhou2018qaoaperformance,crooks2018performance}.   %
From the perspective of lower bounds, a recent paper of Hastings~\cite{hastings2019} demonstrates that for QAOA with $p=1$, classical local algorithms can
achieve the same (or better) performance guarantees as QAOA; the author  argues that for some problems increasing the QAOA depth by any bounded amount would not lead to significant improvements in performance.
Farhi et al.~\cite{farhi2020qaoaneedsfullgraph} show an upper bound for the approximation ratio attainable by a particular realization of QAOA for the MaxIndependentSet problem on $d$-regular graphs with depth~$p$ growing slower than logarithmically with the number of nodes, implying that%
~$p$ must grow at least logarithmically for QAOA to outperform a particular approximation ratio.  
They note the observed relationship between the required depth and the graph diameter, suggesting that for graphs with larger diameter (e.g., 2d lattices) the depth required to obtain a quantum advantage in this case 
may be polynomially scaling in the problem size.

While clearly important, these findings leave many  practical questions unresolved, such as  ``What is the QAOA depth required to adequately solve a given problem instance?''
For example, Farhi et al.~\cite{farhi2020qaoaneedsfullgraph} note that for a class of 3-regular graphs with 2 million qubits, the upper bound they prove  yields a necessary depth of only $p=7$.

A promising approach to overcoming limitations of the previous analyses is the application of symmetry ideas and analysis 
to QAOA.
Understanding symmetry has
long been central to the analyses of topological phases of matter~\cite{chiu2016classification}, the hardness of ground state preparation for practically interesting local Hamiltonians~\cite{hastings2012tlescommutingham,Eldar2017}, mathematical programming~\cite{ostrowski2011orbital,barnhart1998branch}, and quantum query complexity~\cite{2006.12760}, among other important applications in quantum computing.   
Much insight likewise may be gained through 
the connection between the symmetries of the QAOA ansatz (that is, the symmetries
shared by the QAOA %
quantum circuit and initial state) and the symmetries of the target classical optimization problem.
Indeed, recently Bravyi et al.~\cite{braviyobstacles} have utilized the  $\mathbb{Z}_2$ symmetry exhibited by 
the \maxcut{} problem and its corresponding QAOA ansatz %
together with the properties of bipartite expander graphs to show that constant-depth QAOA for %
\maxcut{} %
is outperformed by the classical Goemans--Williamson algorithm for some families of graphs. This result highlights the potential
of %
symmetry insights in advancing our understanding of the performance and limitations of QAOA and related approaches to quantum optimization. 

In this paper we focus on QAOA as originally proposed in~\cite{farhi2014quantum} for clarity of presentation, although our approach and analysis serve as a template that may be adapted to more general QAOA ansatz\"e~\cite{hadfield2017quantum,bartschi2020grover} or other parameterized quantum circuits. %
Our %
main technical result (Theorem~\ref{thm:qaoasym}) connects the classical symmetries of the problem objective (cost) function to be optimized to symmetries of the QAOA dynamics,  %
operationalized in terms of QAOA outcome probabilities. 
\rev{By mapping the symmetries to quantum operators and considering corresponding commutators, w}e show that the binary strings connected by a symmetry of the objective function will have the same output probabilities. %
For the special case of variable (qubit) permutation symmetries, the 
theorem implies the following %
simple corollary. 

\begin{corollary}\label{cor:permenergy}
 Suppose a group of variable permutations leaves the objective function invariant. Then QAOA %
 output probabilities are the same across all bit strings connected by such permutations, for all fixed choices of QAOA parameters and depth.
 \end{corollary}

\noindent More generally, Sec.~\ref{sec:theory} shows that %
symmetries of the classical objective function induce a group of symmetries of the QAOA ansatz,
and this leads to conserved probability amplitudes over the subspace of states connected by these symmetry transformations. 
Our results remain applicable even if we know only  some convenient 
subgroup of the group of all possible symmetries of the classical objective function.
An important and well-known fact is that such symmetries can 
imply quantum dynamics restrict to a reduced Hilbert space. 
For QAOA we quantify the relationship between the group of symmetries and the resulting dimension of the effective Hilbert space in terms of group properties~(Theorem~\ref{thm:qaoareduced}), which in some cases can be exponentially reduced, and we discuss implications to %
classical simulation. We illuminate our results with several examples.  

In Sec.~\ref{sec:learningtopsection} we %
show an illustrative %
application of our %
results 
by extending %
our framework %
to the practical task of predicting the QAOA depth required to achieve a desired approximation ratio, in the setting of linear QAOA parameter schedules. 
We %
demonstrate a machine learning approach to the instance-wise analysis of QAOA that %
constructs symmetry groups and a number of related symmetry measures of the classical optimization problem in order to extract %
features used %
to predict the required QAOA depth. We focus specifically on \maxcut{}, which %
is a standard problem considered for QAOA~\cite{farhi2014quantum,wang2018quantum,hadfield2018quantum,crooks2018performance,zhou2018qaoaperformance,Shaydulin2019EvaluatingDOI,Guerreschi2019,braviyobstacles,Szegedy2019qaoaenergies}, %
with symmetries related to those of the underlying graph, though our approach may be extended to other problems of interest.  
We note that for many %
practically important but %
classically hard classes of problem instances, such as \maxcut{} on bounded degree graphs, the full symmetry group of the graph can always be %
constructed in polynomial time~\cite{luks1982isomorphism}.

The %
paper is organized as follows. Section~\ref{sec:background} presents an overview of the relevant background. Section~\ref{sec:theory} presents a series of analytical results exploring the connections between classical symmetry and QAOA. Section~\ref{sec:learningtopsection} provides an example application %
of %
using symmetry considerations to predict the QAOA depth required achieve a target approximation ratio for \maxcut{} in the case of linear QAOA schedules.
The methods and numerical results in Sec.~\ref{sec:learningtopsection} are mostly self-contained, and the interested reader may proceed there directly, independently of the details and proofs of Sec.~\ref{sec:theory}. Section~\ref{sec:discussion} concludes %
our study with an outlook toward future applications. 

\noindent {\bf Reproducibility: } %
The Python code for the numerical experiments of Sec.~\ref{sec:learningtopsection} and the %
generated dataset are available online~\cite{code,rawdata}.

\section{Background}\label{sec:background}

In this section we briefly review our notions of binary optimization, QAOA, and classical symmetries. Though for simplicity we consider maximization problems 
our results apply similarly to minimization. %

\paragraph{Binary optimization} Consider a class of problem instances, that is,
nonnegative objective functions $f(x)$ on the Boolean cube $ \{0,1\}^n$, with the corresponding 
optimization problem
\begin{equation}\label{eq:binopt}
\max_{x\in \{0,1\}^n}f(x). %
\end{equation}
\noindent 
The objective function~$f$ may be encoded as the $n$-qubit Hamiltonian $C$, %
which acts as $C\ket{x}=f(x)\ket{x}$ for each $x\in\{0,1\}^n$, namely, the diagonal matrix in the computational basis
\[
C=diag(f(x)). %
\]

\noindent Compact representations of such Hamiltonians (in terms of Ising spin, i.e.,  Pauli-$Z$ matrices) can be constructed efficiently for many combinatorial optimization problems~\cite{hadfieldrepresentation}. %
A candidate solution $x^*\in \{0,1\}^n$ %
is %
an $r$-approximation %
for a given instance, $0\leq r\leq 1$ if
\begin{equation}\label{eq:approx_r}
f(x^*) \; \geq \; r \cdot \max_{x}f(x),    
\end{equation}
and an algorithm is said to %
achieve approximation ratio~$r$  
for~(\ref{eq:binopt}) %
if it returns %
an $r$-approximation or better 
for every problem instance in the class (i.e., in the worst case). %
For many applications (e.g., NP-hard problems) the best \textit{efficient} algorithms known can  achieve only some $r<1$, where~$r$ may or may not depend on the problem size~\cite{vazirani2013approximation}. 
Similar definitions apply for minimization problems.

\paragraph{\maxcut{}} %
An objective function can be defined on a graph $G=(V,E)$ by choosing an appropriate encoding. For %
graph problems such as \maxcut{}, the commonly used encoding we consider assigns a binary indicator variable to each vertex %
such that variable assignments %
correspond to subsets of vertices. With this choice of encoding, the set of binary strings $x\in\{0,1\}^{|V|}$ encodes the possible configurations of the graph using~$|V|$ qubits; note that alternative choices of encoding may give different configuration spaces with potentially different resource requirements such as number of bits.
For \maxcut{} the objective 
is find a subset of vertices (i.e., partition the %
vertices~$V$ into two disjoint parts) such that the number of cut graph edges (those with an endpoint in each part) is maximized. %
The \maxcut{} problem on general graphs is APX-complete~\cite{papadimitriou1991optimization}, which means it
has no polynomial-time approximation scheme %
unless P=NP (i.e., 
no efficient deterministic
classical algorithm can find a cut achieving an approximation ratio better than some constant), 
although efficient 
algorithms
exist for particular classes of graphs~\cite{Arora1999}. 
In light of this computational difficulty, \maxcut{} has been considered extensively for QAOA; see, for example, ~\cite{farhi2014quantum,wang2018quantum,braviyobstacles}. 
In the standard problem 
mapping %
the %
objective function %
is encoded by the Hamiltonian $$C_{\text{\maxcut{}}} = \frac{1}{2}\sum_{(u,v)\in E}(I-Z_uZ_v),$$ where $Z_u$ indicates the single-qubit Pauli-$Z$ operator applied to the qubit corresponding to vertex%
~$u$. Although in this paper we consider \maxcut{}, we emphasize that our approach and
results are general and may be applied to other problems of interest. %

\paragraph{QAOA} The Quantum Approximate Optimization Algorithm is a hybrid quantum-classical algorithm that combines a parameterized quantum evolution with a classical outer-loop optimizer to (approximately) solve binary optimization problems~\cite{farhi2014quantum,hadfield2017quantum}.
At each call to the quantum computer, 
a trial state is prepared by applying a sequence of quantum %
operators in alternation 

\begin{equation} \label{eq:QAOAstate}
\ket{\vec{\beta}, \vec{\gamma}}_p := %
U_B(\beta_p)U_C(\gamma_p)\ldots U_B(\beta_1)U_C(\gamma_1)\ket{s},
\end{equation}

\noindent where $C$ is the problem Hamiltonian, $U_C(\gamma) = e^{-i\gamma C}$ is the phase operator, $U_B(\beta)$ is the mixing operator, and $\ket{s}$ is some easy-to-prepare initial state. %
The parameterized quantum circuit (\ref{eq:QAOAstate}) is called the \emph{QAOA ansatz}. We refer to the number of alternating operator pairs~$p$ as the \textit{QAOA depth}. The selected parameters $\vec{\beta},\vec{\gamma}$ are said to define a \emph{schedule}, analogous to a similar choice %
in quantum annealing. For unconstrained optimization problems as we consider in this paper,\footnote{Alternative problem encodings, mixing operators, and initial states for QAOA have been proposed; see e.g.~\cite{hadfield2017quantum}.} the standard mixing operator is $U_B(\beta)=e^{-i\beta\sum_jX_j}$ which corresponds to time evolution under the transverse field Hamiltonian $B:=\sum_j X_j$, and the initial state is the uniform superposition product state $\ket{s} :=\tfrac{1}{\sqrt{2^n}}\sum_{x\in\{0,1\}^n} \ket{x}=\ket{+}^{\otimes n}$. Here $X_j$ similarly denotes the single-qubit Pauli-X operator applied to qubit $j$. 

Preparation of the state (\ref{eq:QAOAstate}) is followed by a measurement in the computational basis; we let $\Prob_p(x)$
denote the probability of obtaining bitstring $x\in \{0,1\}^n$ with 
\[
\Prob_p(x) = |\braket{x|\vec{\beta},\vec{\gamma}}_p|^2.
\]
The output of repeated state preparation and measurement may be used by a classical outer-loop algorithm to select the schedule~$\vec{\beta},\vec{\gamma}$. 
We consider %
optimizing the expectation value $\langle C\rangle_p = \bra{\vec{\beta}, \vec{\gamma}}_pC\ket{\vec{\beta}, \vec{\gamma}}_p$
as originally proposed in~\cite{farhi2014quantum},
although other outer-loop objectives have been %
suggested in the literature; see, for example,~\cite{Barkoutsos2020}. 
The output of the overall procedure is the best solution found for the given problem. 
We emphasize that the task of finding good QAOA
parameters appears challenging in general. In Sec.~\ref{sec:learningtopsection} we 
explore a strategy for reducing the number of %
free parameters towards mitigating this issue. 
Hence we say that depth~$p$ QAOA achieves approximation ratio $r$ for a problem instance $f(x)$ if there exist parameters $\vec{\beta},\vec{\gamma}$ such that 
\begin{equation}\label{eq:qaoa_approx_r}
 \langle C\rangle_p := \bra{\vec{\beta}, \vec{\gamma}}_pC\ket{\vec{\beta}, \vec{\gamma}}_p \; \geq \; r \cdot \max_{x}f(x), 
\end{equation}
and similarly in the case of \textit{QAOA with parameter schedules} where%
~$\vec{\beta},\vec{\gamma}$ satisfying \eqref{eq:qaoa_approx_r} are restricted to a specified set of allowed values as considered in Sec.~\ref{sec:learningtopsection}. 
Note that under this definition the approximation ratio can be evaluated exactly in simulation, but only approximately on a quantum computer due to finite sampling error.

\paragraph{Classical symmetries} We define a \emph{classical (objective function) symmetry} to be a permutation of the Boolean cube
mapping each $x\rightarrow x'$ such that the objective function $f(x)=f(x')$ is %
unchanged, for each $x\in\{0,1\}^n$. More formally, we consider the symmetric group $S_{2^n}$ of permutations on $n$-bit strings~\cite{rotman2015advanced} (i.e., permutations of $2^n$ objects) and its subgroups (called permutation groups), in particular the natural subgroup $S_n\subset S_{2^n}$ induced by the group of
permutations of the bit indices. 
(We use $S_n$ to denote this %
important subgroup throughout.) 
Then a \emph{symmetry} of the objective function $f(x)$ is a transformation $a\in S_{2^n}:\{0,1\}^n\rightarrow \{0,1\}^n$ such that for each $x \in \{0,1\}^n$ we have $f(x)=f(a(x))$. One can easily  see that a set of such symmetries gives a group under composition; indeed, we leverage several ideas from group theory in our results below.

On qubits, each symmetry $a(x)$ %
may be faithfully represented as the permutation matrix \\
$A =\sum_{x\in\{0,1\}^n}\ket{a(x)_1\ldots a(x)_{n}}\bra{x_1\ldots x_n}$,
which acts as $A\ket{x}=\ket{a(x)}$ for each computational basis state $\ket{x}$ and is unitary $A^\dagger A=I$. With this representation in mind, we will sometimes write  $A\in S_{2^n} $, %
in place of the permutation $a$. 
(In particular, we will also use $S_{n}$ to denote %
permutations on $n$ qubits.) 
Then an objective function having symmetry $A\in S_{2^n}$ %
is equivalent to the condition $A^\dagger C A=C$ for its Hamiltonian representation~$C$. For example, this condition is always satisfied for %
the MaxCut Hamiltonian $C_{\text{\maxcut{}}}$ above %
for the symmetry $A=X_1X_2\dots X_n\in S_{2^n}$,
which corresponds to %
flipping all bits simultaneously. %
Similarly, we say a general %
operator $Q$ acting on $n$ qubits is invariant under a symmetry $A$ if it satisfies $A^\dagger QA=Q$.

Any set of objective function symmetries $\{A_1,\dots, A_\ell\}\subset S_{2^n}$ generates a group under composition.
Two bit strings $x,y \in \{0,1\}^n$ are then said to be in the same orbit if there exists a symmetry group element $A$ such that $\ket{y}=A\ket{x}$, which gives an equivalence relation on $\{0,1\}^n$ (or, equivalently, on $\{\ket{x}:x\in\{0,1\}^n$\}). 
We remark that our consideration of permutation groups is  general since from Cayley's theorem~\cite{rotman2015advanced} every finite group is isomorphic to a permutation group on some underlying set. 
Two important observations for our results to follow are that the standard QAOA initial state $\ket{s}=\ket{+}^{\otimes n}$ satisfies $A\ket{s}=\ket{s}$ for any %
$A\in S_{2^n}$ and that for any qubit permutation $A\in S_n$ the transverse-field mixing Hamiltonian is invariant~$A^\dagger BA=B$.

\paragraph{Graph symmetries} 
A symmetry of a graph is a transformation that leaves the graph invariant. %
 For optimization problems on a graph such as \maxcut{} with a given encoding, graph symmetries induce classical symmetries, typically corresponding to a subgroup of the qubit permutations~$S_n$. %
Indeed, we will primarily consider
the well-studied class of graph (vertex) automorphisms, that is,    %
permutations of the vertices that %
map a graph onto itself. Objective functions on graphs that do not depend on the vertex or edge labels are naturally invariant under such transformations.
More formally, for a graph $G=(V,E)$, an  automorphism is a permutation $\sigma:V\rightarrow V$ such that $(\sigma(v), \sigma(u))\in E$  if and only if  $(u,v)\in E$. One can easily  see that the automorphisms of a given graph form a group under composition~\cite{biggs1993algebraic}.
While the size (order) of the group of automorphisms 
can be as large as $n!$ for structured graphs, random graphs generally have much smaller automorphism groups. For example, random three-regular graphs or Erdos--Renyi model graphs almost surely have no nontrivial automorphisms~\cite{kim2002asymmetry,erdHos1963asymmetric}, though practical applications may involve classes of graphs with useful structure; see Sec.~\ref{sec:measuresymmhardness} for 
additional %
discussion. 
Two vertices $u$ and $v$ are said to belong to the same orbit if there exists an automorphism $\sigma$ such that $\sigma(v)=u$, and %
orbits induce equivalence classes on the 
vertex set $V$. For an in-depth presentation %
of graph automorphisms and their properties the reader is referred to%
~\cite[Part Three]{biggs1993algebraic}.

Our results in the next section consider graph symmetries of objective functions, as well as classical symmetries more generally. Our application to \maxcut{} in Sec.~\ref{sec:learningtopsection} considers graph automorphisms as well as several related notions of graph symmetry described therein. %

\section{Classical symmetries in QAOA}\label{sec:theory}

Various general notions of symmetry for quantum algorithms are possible, and different notions may be appropriate in different applications. 
When an underlying classical function is concerned, for example an objective function over a set of configurations on a graph, a natural direction to explore is what its symmetries from the classical perspective 
can tell us about derived quantum algorithms such as QAOA. We consider classical symmetries of the objective function (i.e., problem Hamiltonian) and show how they relate to symmetries of the QAOA ansatz and its resulting measurement outcome probabilities. In particular we consider the special case of graph symmetries.
While in general it %
may not be possible to efficiently obtain the full symmetry group of the objective function, %
in the case of graph symmetries for many (but not all) practically interesting classes of  problems 
there exist polynomial-time algorithms or fast solvers \rev{that lack performance guarantees, but compute the full 
graph automorphism symmetry group quickly in practice}; we elaborate on these considerations in Sec.~\ref{sec:measuresymmhardness}. An appealing aspect of our results is that they apply regardless, i.e., when only some subgroup of classical symmetries %
is known. 

For QAOA, symmetries relate to both the phase and mixing operators (and hence to the problem and mixing Hamiltonians). For classical symmetries shared by both operators,
Theorem~\ref{thm:qaoasym} (and Lemma~\ref{lem:qaoasymgeneral} more generally) shows 
how such symmetries %
restrict  
the QAOA dynamics: %
\emph{%
QAOA amplitudes in the computational basis  are always the same for states connected by the %
the group of such symmetry transformations, and hence likewise for measurement probabilities.}  \rev{In other words,  the classical symmetries of the optimization problem give rise to %
symmetries of the QAOA dynamics.} Our results are general and are applicable to a wide class of optimization problems. %
Note that useful symmetries may arise from the problem encoding, the particular problem class or instance, or the quantum states and operators themselves. Although we focus
in this paper on QAOA with the standard transverse-field mixer and initial state,
our results may be similarly extended to more general mixing operators with potentially different symmetries, such as those described in~\cite{hadfield2017quantum,bartschi2020grover}, %
or to different choices of QAOA initial state.

We now state our main theorem, followed by several useful corollaries,  %
as well as an explicit example of QAOA symmetry for an instance of MaxCut.  The proof of the theorem is given using a more general lemma
presented in Sec.~\ref{sec:lemmas}.

\begin{theorem}\label{thm:qaoasym}
Consider a depth-$p$ QAOA state $\ket{\vec{\beta}, \vec{\gamma}}_p %
= U_B(\beta_p)U_C(\gamma_p)\ldots U_B(\beta_1)U_C(\gamma_1)\ket{s}$ with objective Hamiltonian $C$, transverse-field mixing Hamiltonian $B$, and initial state $\ket{s}=\ket{+}^{\otimes n}$.
\begin{itemize}
    \item %
    Given a classical symmetry matrix $A\in S_{2^n}$ implementing the permutation $A\ket{x}=\ket{a(x)}$ for each  $x\in\{0,1\}^n$ such that  (i) $[A,C]=0$ and (ii)  $[A,B]=0$, then for all $p,\vec{\beta},\vec{\gamma}$, the measurement probability of all bit strings connected by %
    $A$ is the same:   
\begin{equation}
   \forall x\in \{0,1\}^n \;\;\;\;\;\;\;\ \Prob_p(x) = \Prob_p(a(x)).
 \end{equation}  
 \item %
 Suppose we know a number of classical symmetries $A_1,\dots A_\ell \in S_{2^n}$ such that  (i) $[A_j,C]=0$ and (ii) $[A_j,B]=0$ for each $j=1,\dots,\ell$. %
Then taking matrix products of the $A_j$ generates the %
QAOA symmetry group
\begin{equation}
    \mathcal{A}=\{A_0=I,A_1,\dots A_\ell, A_{\ell+1}\dots A_{\ell'}\}\subset S_{2^n},
\end{equation} %
which acts as $A_j\ket{x}=\ket{a_j(x)}$ for each $x\in\{0,1\}^n$, $A_j\in\mathcal{A}$, and
for all $p,\vec{\beta},\vec{\gamma}$, %
\begin{equation}
   \forall x\in \{0,1\}^n, \;\;
    \forall A_j\in \mathcal{A} \;\;\;\;\;\;\;\ \Prob_p(x) = \Prob_p(a_j(x)),
 \end{equation}  
  i.e., probabilities are the same across all strings in the orbit $\mathcal{A}\cdot x$.
\item %
In either case a symmetry $A\in S_{2^n}$ %
satisfies %
$[A,C]=0=[A,B]$ if and only if  (i') $c(a(x))=c(x)$ and (ii') $\{a(x^{(j)}):j=1,\dots,n\}=\{a(x)^{(j)}:j=1,\dots n\}$
as subsets of $\{0,1\}^n$, where $y^{(j)}$ denotes the bitstring $y$ with its $j$th bit flipped. 
\end{itemize}
\end{theorem}

We remark that for the transverse field mixer $B = \sum_j X_j$, qubit permutations $A\in S_n\subset S_{2^n}$ always satisfy 
condition (ii) of the theorem, although $A\in S_n$ is not necessary. 
Corollary~\ref{cor:permenergy} of the Introduction then  follows trivially %
observing that the variable permutation symmetry of the objective function implies qubit permutation symmetry of the objective Hamiltonian.
An example satisfying $[A',B]=0$ but where $A'\notin S_n$ is $A'=X_1X_2\dots X_n\in S_{2^n}$ which flips all bits simultaneously. We emphasize that the theorem may be applied instance-wise or in some cases over a class of problem instances. 
For example, $A'$ applies to all instances of MaxCut, corresponding to the well-known $\mathbb{Z}_2$-symmetry of that problem, with 
additional symmetries arising from the %
 automorphisms of the particular graph instance (see Cor.~\ref{cor:graph} below).

Theorem~\ref{thm:qaoasym} and Corollary~\ref{cor:permenergy} consider %
general objective function symmetries. A practically %
important subset of such symmetries are ones that can be understood as the symmetries of some graph derived from the structure of the objective function and corresponding problem Hamiltonian.
Following for instance~\cite{hastings2012tlescommutingham}, we refer to this graph as the \emph{interaction graph} of a given problem instance. For example, for MaxCut and problems similarly defined over a set of configurations on a graph,
the interaction graph is typically the underlying graph on which the problem is defined.\footnote{Note that in cases where the terms of the objective polynomial do not all have equal coefficients (as they do in the case of, e.g., MaxCut), the relevant symmetries of the interaction graph are \emph{weighted} automorphisms (i.e., automorphisms preserving edge weights~\cite{Balasubramanian1994}). Obtaining the group of weighted automorphisms is in general more difficult than in the nonweighted case, and the introduction of weights makes the problem less likely to have nontrivial automorphisms; hence we do not consider them in further detail.} 
For an arbitrary polynomial objective function defined on the Boolean cube, we can define an interaction graph by fixing variable labeling, letting each variable be a vertex, and connecting two vertices by an edge if there is a term in the polynomial including the corresponding variables. %
For an arbitrary pseudo-Boolean function (including constraint satisfaction problems, commonly considered in QAOA) such polynomial representations exist and can be constructed by considering the function's Fourier expansion~\cite{hadfieldrepresentation}. %
Hence, Cor.~\ref{cor:permenergy} implies that symmetries $A\in S_n$ of the %
interaction graph %
always yield QAOA symmetries, although as explained not all symmetries of the objective function (problem Hamiltonian) are graph symmetries in general. We formalize this %
for the case of problems on graphs. 

\begin{corollary}\label{cor:graph}
Consider an $n$-vertex graph $G$ %
with some (label-independent) objective function $f$ defined on its vertex configurations~$\{0,1\}^n$. Suppose we know a subgroup of graph automorphisms $\mathcal{A}_G\subseteq \textrm{Aut}(G)\subseteq S_n\subset S_{2^n}$.  Then for all $p,\vec{\beta},\vec{\gamma}$ the corresponding QAOA states satisfy 
 \begin{equation}
   \forall x\in \{0,1\}^n \;\;
    \forall a_j\in\, \mathcal{A}_G \;\;\;\;\;\;\;\ \Prob_p(x) = \Prob_p(a_j(x)).
 \end{equation}  
Generally, $\textrm{Aut}(G)$ may be a proper subgroup of the full QAOA symmetry group $\mathcal{A}$.
\end{corollary}
\begin{proof}
Each permutation matrix $A\in S_n$ representing a graph automorphism satisfies the conditions of the theorem %
as in Corollary~\ref{cor:permenergy}. As mentioned the \maxcut{} symmetry $A'=X_1X_2\dots X_n$ of flipping all bits shows not all symmetries correspond to graph automorphisms. 
\end{proof}

\paragraph{Example: QAOA for \maxcut{} on $K_{3,3}$}
\begin{figure}[h]
\subfloat[$100000$]{\includegraphics[width=0.23\textwidth]{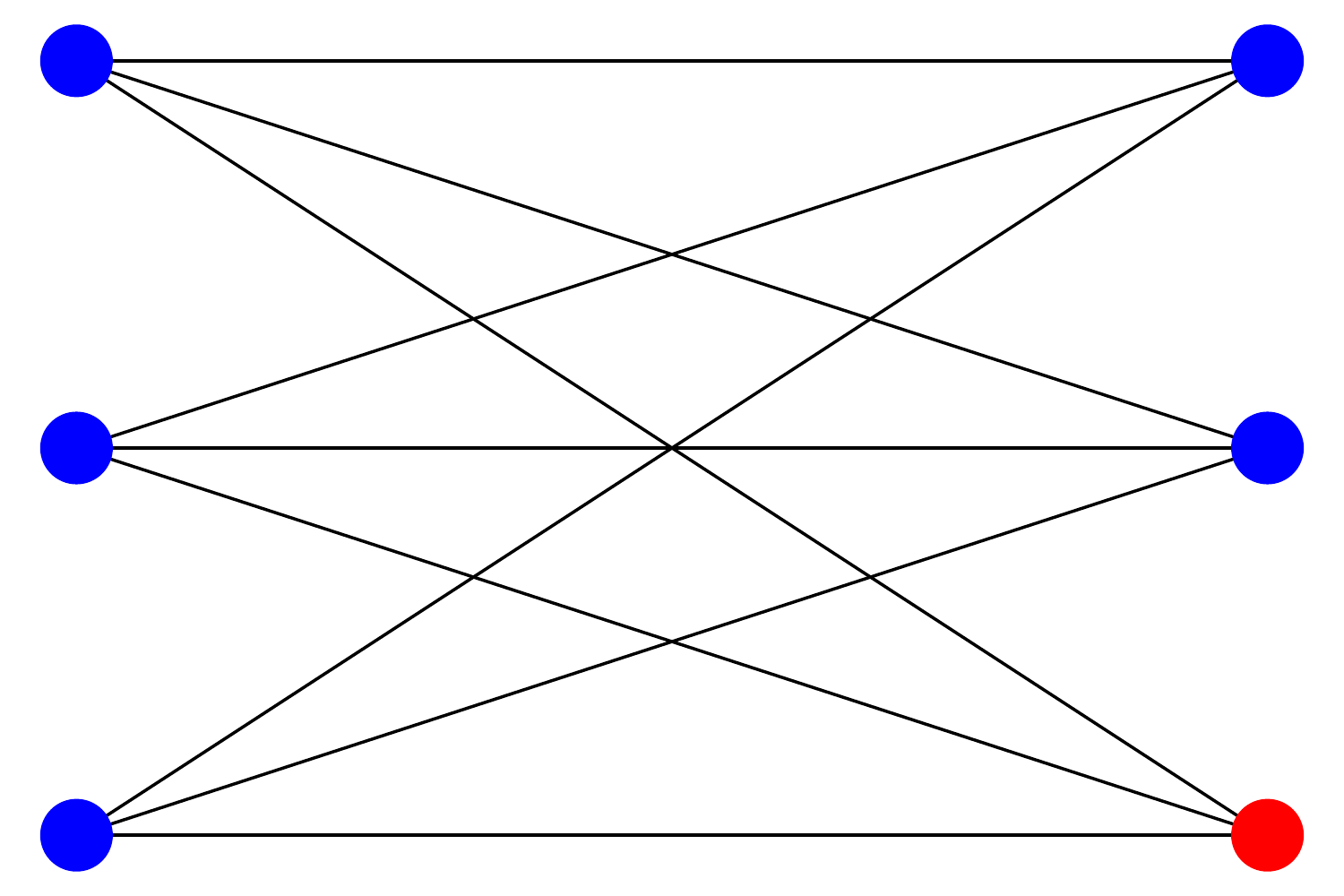}%
\label{fig:k33ham1n0}}
\hfill
\subfloat[$000001$]{\includegraphics[width=0.23\textwidth]{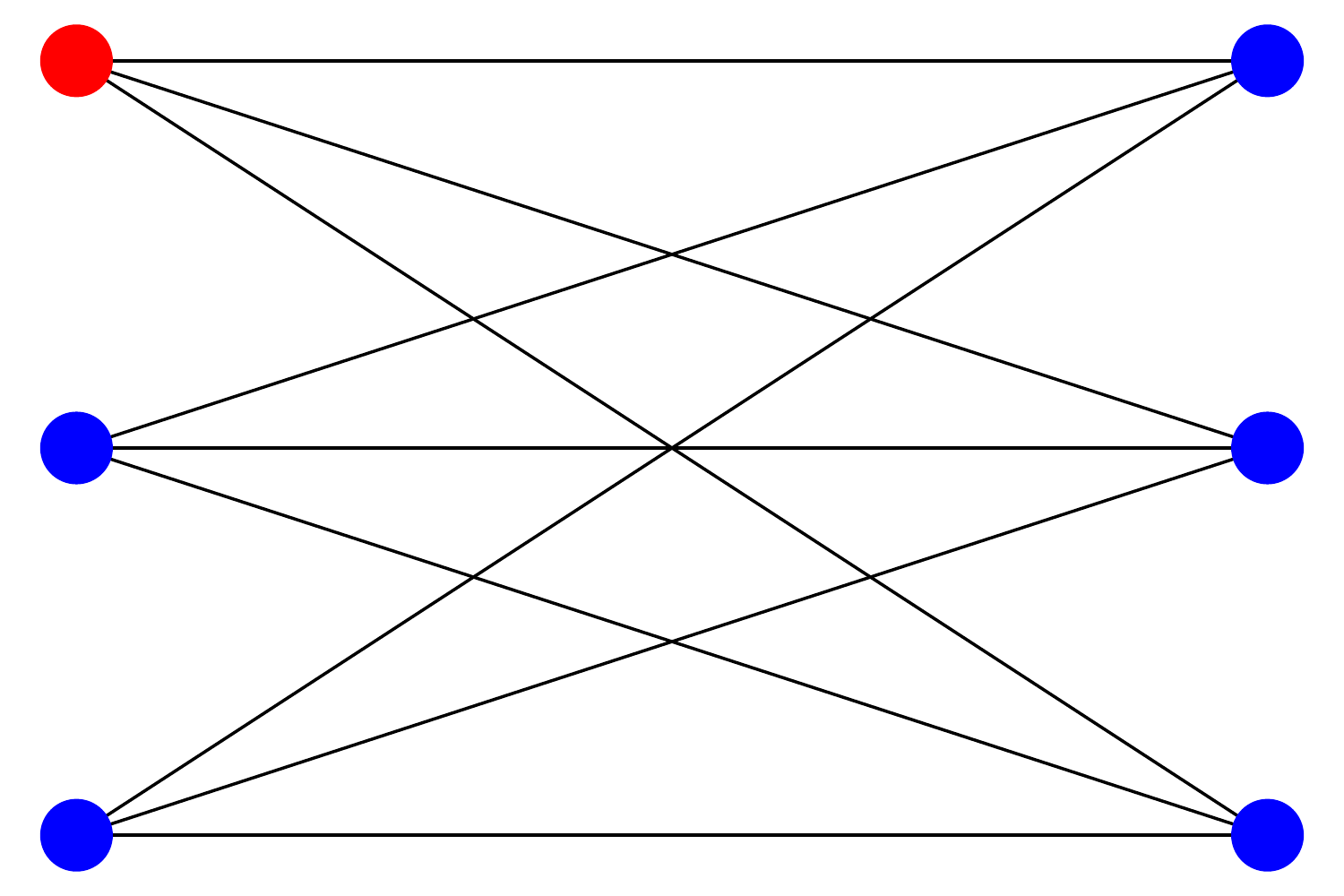}%
\label{fig:k33ham1n5}}
\hfill
\subfloat[$000111$]{\includegraphics[width=0.23\textwidth]{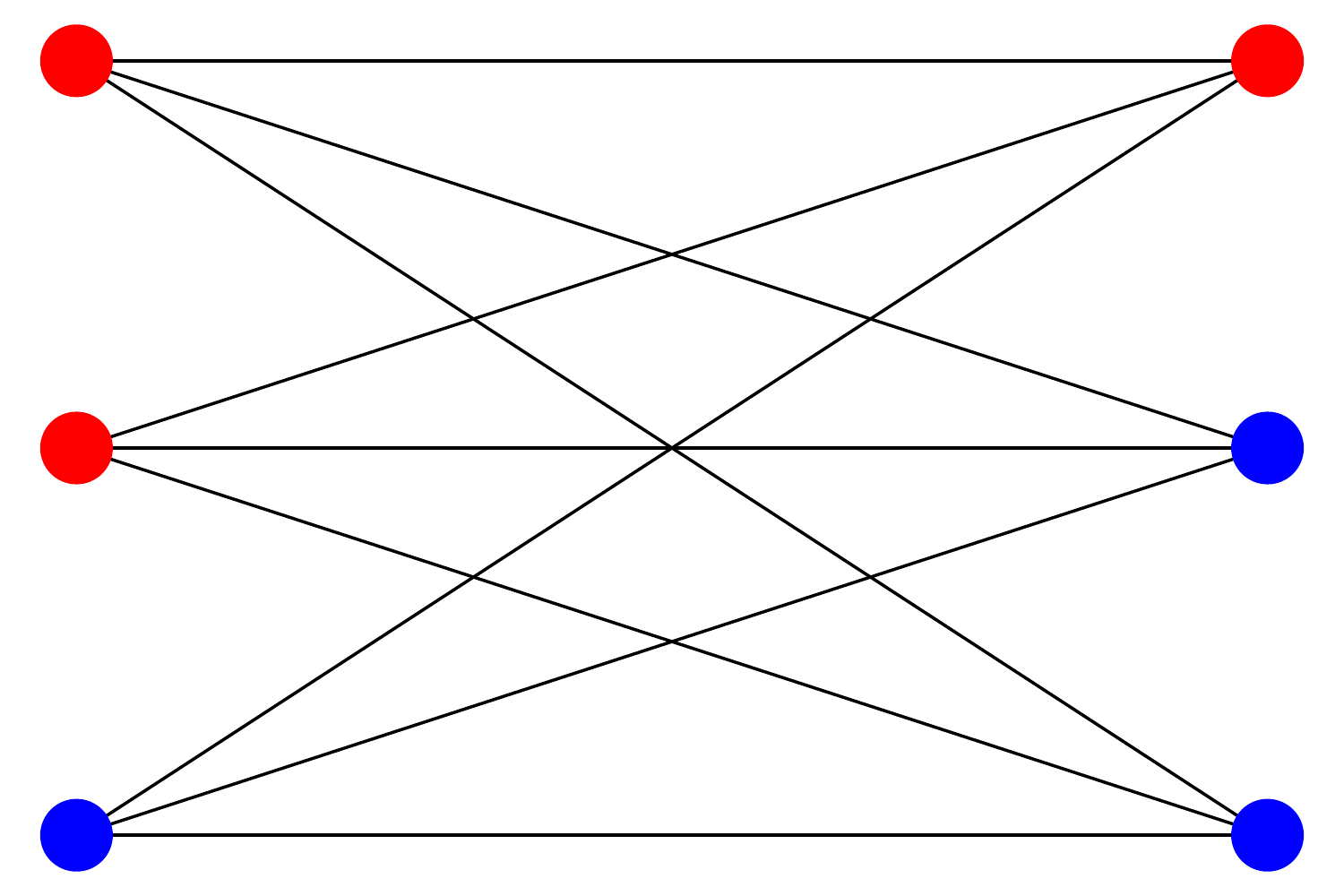}%
\label{fig:k33hamm3n1}}
\hfill
\subfloat[$010101$]{\includegraphics[width=0.23\textwidth]{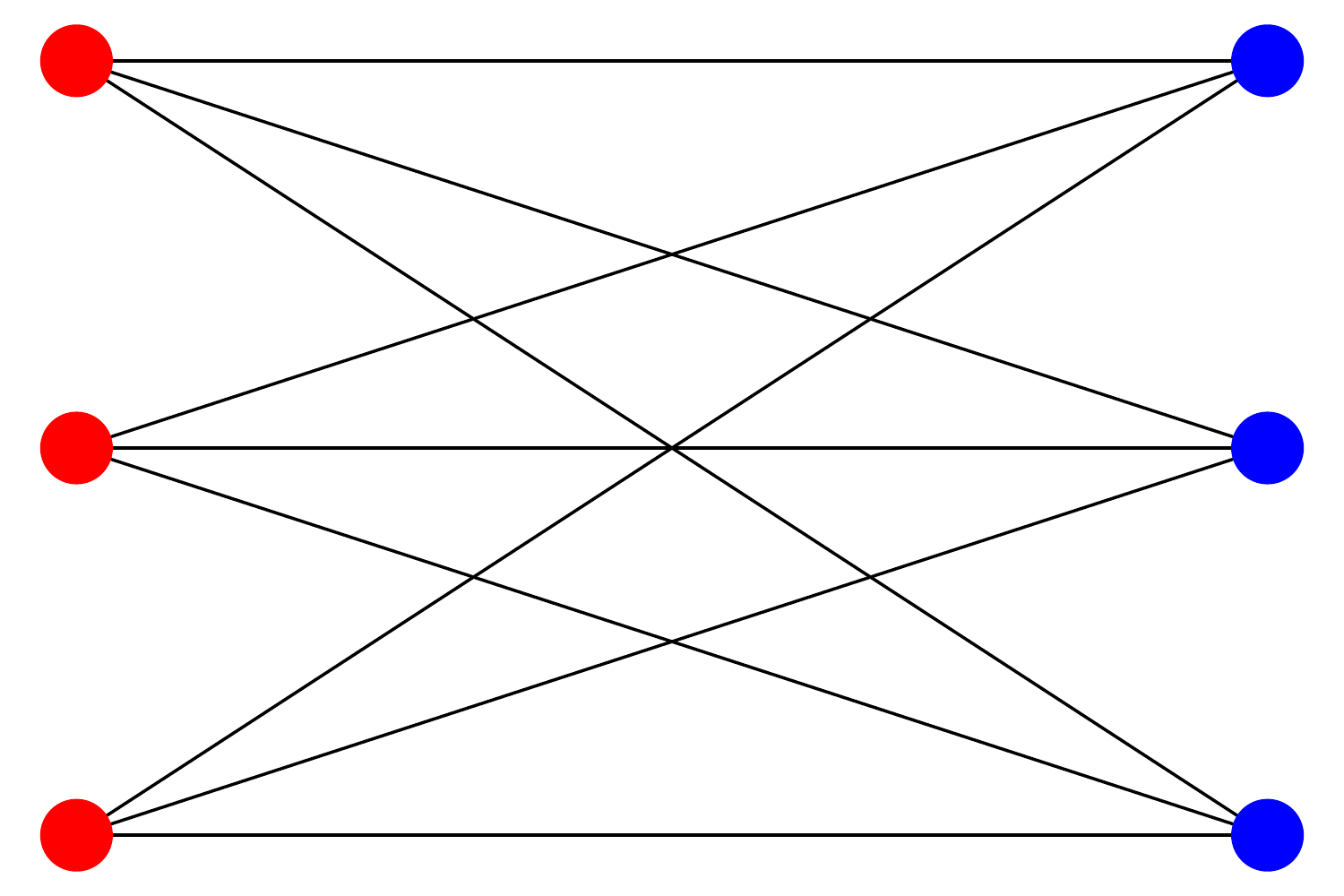}%
\label{fig:k33hamm3n2}}
\caption{\maxcut{} on $K_{3,3}$ graph: All nodes of the graph are on the same vertex orbit in $V$ (i.e., there exists an automorphism that takes one into another; for example, Fig.~\ref{fig:k33ham1n0} can be taken into Fig.~\ref{fig:k33ham1n5} by a vertical mirror symmetry applied after horizontal mirror symmetry).
Corollary~\ref{cor:graph} implies all bitstrings with Hamming weight $1$ have the same QAOA measurement probabilities. %
On the other hand, no graph automorphism  takes the red nodes in Fig.~\ref{fig:k33hamm3n1} into those of  Fig.~\ref{fig:k33hamm3n2}. %
}\label{fig:k33}
\end{figure}

To illustrate %
our results consider %
\maxcut{} on the 6-node complete bipartite graph $K_{3,3}$. This graph has a single vertex orbit (see Fig.~\ref{fig:k33}). Recall that vertex orbits induce orbits on bitstrings. Hence from Corollary~\ref{cor:graph} all bitstrings with  Hamming weight
 one will belong to the same orbit in $\{0,1\}^n$ and therefore have the same probabilities in QAOA state $\ket{\vec{\beta}, \vec{\gamma}}_p$, regardless of depth $p$ or the choice of parameters $\vec{\beta}, \vec{\gamma}$. %
At the same time, there does not exist a graph automorphism taking the %
assignment %
$000111$ to
$010101$, as may be anticipated by observing that Fig.~\ref{fig:k33hamm3n1} corresponds to fewer cut edges than does Fig.~\ref{fig:k33hamm3n2}, and so QAOA may yield different measurement probabilities for such pairs of bitstrings.

\subsection{Symmetry preserves probability amplitudes}
\label{sec:lemmas} %
More generally than the measurement probability viewpoint \rev{for QAOA} of Theorem~\ref{thm:qaoasym}, %
we show \rev{%
conditions sufficient for %
classical symmetries to lead to} 
identical probability amplitudes of %
computational basis states. \rev{Specifically, we show that for the %
probability amplitudes of computational basis states connected by a symmetry to be equal, it is sufficient for the symmetry operator to commute with the alternating operators and contain the initial state in its $+1$ eigenspace.}

\begin{lemma}\label{lem:qaoasymgeneral}
Consider depth-$p$ QAOA $\ket{\vec{\beta}, \vec{\gamma}}_p = U_B(\beta_p)U_C(\gamma_p)\ldots U_B(\beta_1)U_C(\gamma_1)\ket{s}$ with classical problem Hamiltonian $C$, mixing Hamiltonian $B$, and initial state $\ket{s}=\ket{+}^{\otimes n}$.

Suppose we know a symmetry  $A\in S_{2^n}$ 
acting as $A\ket{x} = \ket{a(x)}$, $x\in\{0,1\}^n$, 
such that
\begin{enumerate}[i)]
    \item $[A,U_C(\gamma)]=0, \;\forall\gamma\in R_{\gamma}\subset \mathbb{R}$ and
    \item $[A,U_B(\beta)]=0, \;\forall\beta\in R_{\beta}\subset \mathbb{R}$.
\end{enumerate}
Here $R_{\gamma}$ and $R_{\beta}$ are some (possibly discrete) sets of parameter values. Then solution probability amplitudes are invariant under $A$; that is,  for any $x\in\{0,1\}^n$, for all $p$ we have
\begin{equation}
    \braket{x | \vec{\beta}, \vec{\gamma}}_p = \braket{a(x) | \vec{\beta}, \vec{\gamma}}_p, \qquad \forall \vec{\gamma}\in R_{\gamma}^p, \vec{\beta}\in R_{\beta}^p,
   \end{equation}
 and hence, moreover, $\Prob_p(x) = \Prob_p(a(x))$. 
\end{lemma}
\begin{proof}
Observe that since $A$ is unitary, conjugating \textit{i)} by \rev{$A^\dagger$} gives
$$A^\dagger[A,U_C(\gamma)]A^\dagger=\rev{A^\dagger(AU_C(\gamma) - U_C(\gamma)A)A^\dagger =U_C(\gamma)A^\dagger- A^\dagger U_C(\gamma)=}[A^\dagger,U_C(\gamma)]=0$$ 
for all $\gamma\in R_{\gamma}$, and similarly \textit{ii)} implies
$$\rev{A^\dagger[A,U_B(\beta)]A^\dagger=}[A^\dagger,U_B(\beta)]=0$$
for all $\beta\in R_{\beta}$\rev{, i.e., $A^\dagger$ also commutes with $U_B(\beta)$ and $U_C(\gamma)$.} 
Computing QAOA amplitudes %
using the assumptions gives %
\[
\begin{aligned}
 \braket{a(x) | \vec{\beta}, \vec{\gamma}} & = %
 \bra{x}A^{\dagger} U_B(\beta_p)U_C(\gamma_p)\ldots U_B(\beta_1)U_C(\gamma_1)\ket{s} \\
 & = \bra{x}U_B(\beta_p)U_C(\gamma_p)\ldots U_B(\beta_1)U_C(\gamma_1)A^{\dagger}\ket{s} \,\,= \braket{x | \vec{\beta}, \vec{\gamma}}\\ 
\end{aligned}
\]
for any $\gamma_i\in R_{\gamma}$, $\beta_i\in R_{\beta}$, $i = 1,\ldots, p$, as claimed. \rev{Here we have used the commutation assumptions for the first equality and $A^\dagger\ket{s}=\ket{s}$ for the second one.} 
\rev{Following Born's rule,} taking the absolute value squared of each side gives the probability result. 
\end{proof}

The conditions of Lemma~\ref{lem:qaoasymgeneral} are slightly more general than those of Theorem~\ref{thm:qaoasym}; in particular, the conditions $[A,C]=0$ and $[A,B]=0$ of the theorem imply $[A,U_C(\gamma)]=0$ and $[A,U_B(\beta)]=0$ for all $\gamma,\beta\in\mathbb{R}$. 
Using the lemma we now give the proof of Theorem~\ref{thm:qaoasym}. %

\begin{proof}[Proof of Theorem~\ref{thm:qaoasym}.]
We prove the more general second case, which suffices to prove the first.  
First observe that the conditions \textit{i)} and \textit{ii)} of the theorem are sufficient to satisfy those of Lemma~\ref{lem:qaoasymgeneral}; this is easily seen by  expanding the matrix exponential as a power series and using the linearity of the commutator.  
Next observe, trivially, that composing any two permutations gives a permutation, and that 
since $S_{2^n}$ is a group of finite order, iterating any fixed permutation eventually results in the identity. 
It follows that given any subset of $S_{2^n}$, composing elements %
yields a subgroup (if not the whole group). Moreover, if two symmetry operators $A_k,A_l$ satisfy the conditions of Lemma~\ref{lem:qaoasymgeneral}, then clearly so does their product $A_kA_l$. 
Therefore, given any such set of symmetry operators $A_1,A_2,\dots$, the conditions of Lemma~\ref{lem:qaoasymgeneral} are satisfied for each  element of the group generated by these operators, which implies the statements of %
the first two cases of the theorem. 

For the final point of the theorem,
using the unitarity of $A$, we have $[A,C]=0$ if and only if $A^\dagger C A\ket{x}=C\ket{x}$ for all $x\in\{0,1\}^n$, or equivalently $c(a(x))\ket{x}=c(x)\ket{x}$, and hence true if and only if we have $(i')$ $c(x)=c(a(x))$ for each $x\in\{0,1\}^n$. 
Similarly, $[A,B]=0$ if and only if $A^\dagger B A\ket{x}=B\ket{x}$ for each $x\in\{0,1\}^n$. Using that~$B$ acts as on computational basis states as $B\ket{y}=\sum_{j=1}^n\ket{y^{(j)}}$ gives 
\begin{eqnarray*}
A^\dagger B A\ket{x} = A^\dagger \sum_{j=1}^n \ket{a(x)^{(j)}} 
=  \sum_{j=1}^n \ket{a^{-1}(a(x)^{(j)})}. 
\end{eqnarray*}
Therefore, for $A^\dagger B A\ket{x}=B\ket{x}$ to hold, we must have 
$ \{x^{(j)}:j\in[n]\}=\{a^{-1}(a(x)^{(j)}):j\in[n]\}$
as sets of size $n$. Applying $a(\cdot)$ to both sets gives %
$(ii')$, which proves the claim.
\end{proof}

\begin{remark}
Since $p$ is arbitrary in both Lemma~\ref{lem:qaoasymgeneral} and Theorem~\ref{thm:qaoasym}, the results therein also apply to all intermediate states between QAOA layers.
\end{remark}

\subsection{Reduced dynamics from symmetry and applications to classical simulation}
\label{sec:appsimulation}

In principle, each known QAOA symmetry %
allows us to reduce the size of the Hilbert space necessary to exactly simulate QAOA. If a large enough symmetry group is known, the dimension of the reduced subspace may become only polynomially large in the problem input size, in which case efficient classical simulation may be possible, assuming one can efficiently compute the necessary matrix elements. In this section we formalize this observation.

We first recall some additional %
basic group theory beyond that of Sec.~\ref{sec:background}; see~\cite{rotman2015advanced} for a detailed overview.  
Let $\mathbb{B} := \{0,1\}^n$ denote the Boolean cube that is identified with the set of computational basis states $\{\ket{x}:x\in \mathbb{B}\}$, and consider a QAOA symmetry group $\mathcal{A}\subset S_{2^n}$ as in  Theorem~\ref{thm:qaoasym} (i.e., for some fixed problem and objective Hamiltonian~$C$). Recall that the action of $\mathcal{A}$ induces equivalence classes on $\mathbb{B}$, where each class includes all elements $x\in\mathbb{B}$ belonging to the same orbit, denoted $\mathcal{A}\cdot x$ for a particular~$x$.  %
The quotient $\mathbb{B} / \mathcal{A}$ %
gives the set of equivalence classes, %
called coinvariants of $\mathcal{A}$, with $|\mathbb{B} / \mathcal{A}|$ the number of disjoint orbits. Invariants, on the other hand, are fixed points, %
defined for each $A\in\mathcal{A}$ as $\mathbb{B}^{A}:=\{x \in \mathbb{B}: A \cdot x=x\}$. %
Similarly, for each $x\in\mathbb{B}$ its stabilizer subgroup is $\mathcal{A}_{x}:=\{A \in \mathcal{A}: A \cdot x=x\}$. 

We now %
state our result that with symmetry QAOA dynamics can be reduced to a Hilbert space of dimension $|\mathbb{B} /\mathcal{A}|$, which can be expressed as the average number of $x \in \mathbb{B}$ fixed by 
each $A \in \mathcal{A}$, or equivalently the average size of the stabilizer subgroups. 

\begin{theorem}\label{thm:qaoareduced}
Consider a QAOA symmetry group $\mathcal{A}\subset S_{2^n}$ %
satisfying the conditions of Theorem~\ref{thm:qaoasym}. Then QAOA dynamics reduce %
to a subspace of dimension $|\mathbb{B} /\mathcal{A}|$ satisfying  %
\begin{equation}\label{eq:thmqaoasim}
    \left|\mathbb{B} /\mathcal{A}\right|\,=\,\frac{1}{|\mathcal{A}|} \sum_{A \in \mathcal{A}}\left|\mathbb{B}^{A}\right|\,=\,\frac{1}{|\mathcal{A}|} \sum_{x \in \mathbb{B}}\left|\mathcal{A}_{x}\right|\,=\,\sum_{x \in \mathbb{B}} \frac{1}{|\mathcal{A}\cdot x|}.
\end{equation}  
\end{theorem}
Here the three equalities of (\ref{eq:thmqaoasim}) are stated for convenience, since a particular form may be easier to apply to a particular symmetry group. We consider several examples below. 
\begin{proof} We again make use of the isomorphism between the Boolean cube and computational basis states. 
From Lemma~\ref{lem:qaoasymgeneral}, all %
basis states within the same orbit
$\mathcal{A} \cdot|x\rangle:=\{A|x\rangle: A \in \mathcal{A}\}$ %
have the same probability amplitude at all steps of the QAOA algorithm for all possible algorithm parameters, and hence it suffices to track a single amplitude for each orbit. %
For the cardinality, since $|\mathcal{A}|,|\mathbb{B}|$ are finite,  
the first equality of (\ref{eq:thmqaoasim}) is %
Burnside's lemma~\cite{rotman2015advanced}, the second equality follows observing  $\sum_{A\in \mathcal{A}} |\mathbb{B}^A|= |\{(A, x) \in \mathcal{A} \times \mathbb{B}: A \cdot x = x\}|=\sum_{x\in \mathbb{B}} |\mathcal{A}_x|$, %
and the final equality follows observing  $|\mathcal{A}_x| =|\mathcal{A}|/|\mathcal{A}\cdot x|$ for each $x \in \mathbb{B}$ %
from Lagrange's theorem 
and the orbit-stabilizer theorem. 
\end{proof}

For example, consider a fixed $x\in \mathbb{B}$. If $\mathcal{A}\cdot x = \mathbb{B}$, that is, the orbit of $x$ gives all bitstrings, then the third equality of (\ref{eq:thmqaoasim}) implies $|\mathbb{B} /\mathcal{A} |%
=|\mathbb{B}|/|\mathbb{B}|=1$, and so QAOA reduces to one-dimensional (i.e., trivial) dynamics. Indeed, in this case $\mathcal{A}=S_{2^n}$, which implies that the objective function (objective Hamiltonian) is constant. 
If instead the objective function depends only on the Hamming weight so that $\mathcal{A}=S_{n}$, the group of qubit permutations, then clearly the orbit of each $x$ contains all strings of the same Hamming weight, and so $|\mathbb{B} /\mathcal{A}|=n+1$.
Clearly, \lq\lq too much\rq\rq\ symmetry may correspond to easier instances for classical algorithms or even make the problem trivial to solve, as these examples illustrate.
Conversely, for problems without any (known) classical symmetries, we have the trivial symmetry group $\mathcal{A}=\{I\}$, and hence $|\mathbb{B} /\mathcal{A}|=2^n$; in other words, there is no implied reduction in the required Hilbert space size. We consider several less-trivial examples related to MaxCut below.  

As implied by the first two examples above, for problems with sufficient classical symmetry the dimension of the effective Hilbert space may be reduced from $2^n$ to merely a polynomial in $n$. We emphasize, however,  that such a reduction does not necessarily entail classical simulability of the corresponding QAOA circuit.\footnote{%
General efficient classical simulation of even depth $1$ QAOA circuits is impossible under standard computational complexity assumptions~\cite{farhi2016quantum} 
(including, of course, worst-case problem instances with little or no symmetry).} 
In order %
for such symmetry reductions to imply efficient classical simulaton, 
one must know or be able to compute %
various specific details, %
such as the %
relevant matrix elements.%
\footnote{A useful example for comparison is Grover's quantum algorithm for unstructured search; while it is well known that similar considerations reduce the quantum dynamics to a $2$-dimensional Hilbert space, this insight does not allow one to obtain the solution (without the quantum computer). The algorithm's quadratic speedup over classical computers is typically proven by using this reduction. A similar observation is made for a generalization of QAOA with mixing operator derived from Grover's algorithm in~\cite{bartschi2020grover}. }
While this is often the case for toy problems of interest (for example, the Hamming weight ramp problem~\cite{bapat2019hammingramp}), we leave as a topic of future work a detailed  application of Theorem~\ref{thm:qaoareduced} towards deriving more rigorous %
conditions and settings allowing for efficient classical simulation of QAOA circuits. %

Indeed, such investigations may benefit from more technical generalizations of Theorems~\ref{thm:qaoasym} and~\ref{thm:qaoareduced}. We emphasize that the statements and proofs of both theorems depend %
in an essential way on the particular choice of QAOA initial state as the uniform superposition~$\ket{s}=\ket{+}^{\otimes n}$, which guarantees that all of the computational basis states in each symmetry group orbit always have the same probability amplitudes.\footnote{More precisely, this particular state projects entirely to the one-dimension representative (symmetric superposition) corresponding to each element of the quotient space $\mathbb{B}/\mathcal{A}$.} Both theorems may be generalized to arbitrary QAOA initial states, but at the expense of a more complicated presentation in terms of group representation theory beyond the scope of this study. 
Similar considerations apply for extending our results to different mixing operators beyond the %
transverse-field QAOA mixing Hamiltonian.

We conclude the section with three illustrative examples. 
\paragraph{Example: QAOA for problems with $\mathbb{Z}_2$ symmetry} Consider an %
objective function that is invariant under the $\mathbb{Z}_2$ symmetry~$A=X_1X_2\dots X_n$ of flipping all bits simultaneously. We recall, for example, that MaxCut %
always has this symmetry; %
other problems may have this symmetry by construction or only for particular instances. For the symmetry group $\{I,A\}$ clearly we have $\mathbb{B}^{I}=\mathbb{B}$ and $\mathbb{B}^{A}=\emptyset$, and so the first equation of Theorem~\ref{thm:qaoareduced} gives $|\mathbb{B}/\mathcal{A}|=\frac12(2^n+0)=2^{n-1}$; that is, as one may expect, the Hilbert space dimension may be reduced in half. In this case it suffices to consider the set of basis states $\{\tfrac1{\sqrt{2}}(\ket{x}+A\ket{x})\}$ given by the symmetric linear combinations of the two computational basis states in each orbit. 

\paragraph{Example: QAOA for MaxCut on general graphs}
For an arbitrary \maxcut{} instance on a graph $G$ with automorphism group~$Aut(G)$, Theorem~\ref{thm:qaoareduced} implies that the QAOA dynamics reduce to a subspace of dimension $|\mathbb{B}/Aut(G)|$, with basis states given by symmetric sums of the states in each orbit $\mathbb{B}^{Aut(G)}$, or similarly of  dimension $|\mathbb{B}/\mathcal{A}|$ if only a subgroup $\mathcal{A}\subset Aut(G)$ is known. 

Further consider that MaxCut also satisfies the $\mathbb{Z}_2$-symmetry~$A=X_1X_2\dots X_n$ of the preceding example. For each $A_j \in \mathcal{A}\subset Aut(G)$ one can easily  see that $AA_j=A_jA$, and hence  $\mathcal{A}\cup\{A\}$ generates the subgroup $\mathcal{A}'=\mathcal{A}\cup A\mathcal{A}=\{A_1=I,A_2,\dots,A_{|\mathcal{A}|},A,AA_2,\dots,AA_{|\mathcal{A}|}\}\subset S_{2^n}$ of order $2|\mathcal{A}|$, which reduces the effective Hilbert space dimension by a further factor of 2.  

\paragraph{Example: QAOA for MaxCut on complete graphs} Consider MaxCut on the complete graph $G=K_n$ with $n$ vertices and $\binom{n}{2}$ edges. 
Clearly, for complete graphs the MaxCut objective function $f(x)$ is invariant under permutations of the vertices, corresponding to $S_n$ permutation symmetry of the qubit encoding, and the dynamics can again be reduced to an $(n+1)$-dimensional subspace.  
Indeed, a simple calculation shows that the 
objective function may be written as $f(x)=f(|x|)= n|x|-|x|^2$ %
where $|\cdot|$ is the Hamming weight. Let $D$ denote the Hamiltonian representation of the Hamming weight function, that is, $D|x\rangle=|x||x\rangle$ for all $x \in \mathbb{B}$, which is given by the $1$-local Hamiltonian $D=\frac{n}2I-\sum_{j=1}^nZ_j$. 
Then the objective Hamiltonian may be written $C=nD-D^2$, which acts diagonally on the Hamming basis of states $|d\rangle=\tfrac{1}{\sqrt{{n \choose d}}}\sum_{x:|x|=d}|x\rangle$, $d=0,1,\dots,n$. The initial state $\ket{s}$ is easily expressed in this basis as 
\[
|s\rangle %
=\sum_{d=0}^{n} \sqrt{\frac{{n \choose d}}{2^{n}}}|d\rangle,
\]
\noindent %
and a straightforward %
calculation shows the action of the mixing Hamiltonian to be 
\[
B|d\rangle=\sqrt{d(n-d+1)}|d-1\rangle+\sqrt{(d+1)(n-d)}|d+1\rangle,
\]
with the understanding that $\ket{n+1}=\ket{0-1}=0$. 

Hence, the classical simulation of QAOA for MaxCut on complete graphs requires dealing with matrices of size $(n+1)\times(n+1)$, which can be accomplished in poly$(n)$ time for QAOA depth up to $p=$poly$(n)$, modulo accuracy considerations.   
Note that we may again apply the $\mathbb{Z}_2$-symmetry above to further reduce the subspace dimension by an additional factor of $2$ by the identification of pairs of states $\ket{d}$ and $\ket{\bar{d}}:=\ket{n-d}$ (more precisely, to dimension $(n+1)/2$ for $n$ odd or $n/2+1$ for $n$ even).

\subsection{On the difficulty of obtaining the symmetry group}\label{sec:measuresymmhardness}

For some problems 
efficiently determining all  the 
classical symmetries in general may not be possible. 
Nevertheless,  %
our results above apply for any known group of classical symmetries, which may be obtained through a number of possible  approaches. 
For determining graph symmetries, %
efficient algorithms or methods \rev{with worst-case exponential, but sufficiently low to be useful in practice, running time} exist %
for many practically important classes of problem graphs. %
Here we briefly review some %
such approaches %
as well as the complexity of obtaining the full 
graph automorphism group. %
\revtwo{We note that %
graph symmetry analysis techniques 
have been previously applied in the context of characterizing the power of %
adiabatic quantum computation~\cite{Bringewatt2020}.}
Similar considerations apply more generally.

Given an arbitrary graph %
(e.g., an instance of MaxCut), %
an important open problem is whether there exists an efficient classical algorithm for determining its automorphism group (in general), with the best algorithm known resulting from Babai's recent breakthrough discovery of a quasi-polynomial time algorithm for the graph isomorphism problem~\cite{babai2015giquasipoly}. %
Since %
determining a generating set of the graph automorphism group is polynomial time equivalent to %
graph isomorphism~\cite{arvindalgebraandcomputation}, this leads to a quasi-polynomial %
upper bound to the time complexity of determining the full group of automorphisms. %
Nevertheless, 
polynomial-time algorithms exist for many important classes of graphs, %
including those with bounded degree~\cite{luks1982isomorphism}, bounded genus~\cite{filotti1980polynomial,miller1980isomorphism}, or bounded tree-width~\cite{bodlaender1990polynomial} or characterized by particular forbidden minors%
~\cite{grohe2012fixed,grohe2012structural}, among others~\cite{mckay1981practical,McKay201494}.

Note that while the full graph automorphism group $Aut(G)$ can contain up to $n!$ elements, for the applications presented in Sec.~\ref{sec:learningtopsection} it %
suffices to construct its generating set.
Since for an arbitrary group with $m$ elements there always exists a generating set of size $\log{m}$~\cite{arvindalgebraandcomputation}, 
the automorphism group can always be represented efficiently 
(though, as indicated, finding this representation may be challenging). These observations imply that the developed approach can be efficiently %
applied to a variety of important problems.
For example, consider \maxcut{} on three-regular graphs, which is NP-hard to approximate beyond 0.997~\cite{Berman1999}; in this case, the generating set of %
automorphisms %
can be constructed efficiently, and hence the %
techniques of this section and the following 
may be applied.

We emphasize that in practice one often can determine the automorphism group for problems of %
substantial size, although worst-case guarantees may be lacking. Indeed, state-of-the-art \rev{solvers} are able to compute the 
automorphism group generators 
for random graphs with up to tens of thousands of nodes in less then a second~\cite{McKay201494}. \rev{The perfomance of such methods varies with the type of the graph. For example, for vertex-transitive hypercubes, the full automorphism group can be obtained in under 600 seconds even for graphs with millions of nodes~\cite{McKay201494}. However, there are classes of graphs that lead to large search trees and correspondingly larger runtimes even for instances with a few thousands of nodes~\cite{piperno2008search}.}
Many such \rev{methods} have been proposed to accelerate computation of automorphisms, including \texttt{nauty}~\cite{mckay1978computing,mckay1981practical} (which we employ for the application of Sec.~\ref{sec:learningtopsection}), \texttt{Traces}~\cite{McKay201494}, \texttt{saucy}~\cite{darga2004exploiting}, \texttt{Bliss}~\cite{junttila2007engineering}, and \texttt{conauto}~\cite{lopez2014novel}, %
with accelerated running time resulting from the application of
the Weisfeiler--Lehman test and %
search tree pruning heuristics to generate canonical labeling of nodes, among other techniques.

\section{Using symmetry to predict QAOA performance}\label{sec:learningtopsection}

We now present an example application that extends and applies 
the ideas introduced in Sec.~\ref{sec:theory} to study 
the empirical relationship between QAOA performance and symmetries of the objective function in a particular setting. 
We use a machine learning approach employing features computed from the group of automorphisms of a graph, as well as other symmetry considerations, to predict the minimum depth required for QAOA to achieve a fixed approximation ratio for the MaxCut problem on it. 
We focus on studying the connection between symmetry properties of the particular problem instance and QAOA performance. 
To this end, we train two support vector models that  use 
a set of problem symmetry measures. We restrict QAOA parameters to a specific class of schedules, defined and motivated in Sec.~\ref{sec:smoothschedules}, to circumvent %
any computational difficulties related to parameter optimization.
We use symmetry measures corresponding both to exact symmetries as defined in Theorem~\ref{thm:qaoasym}, and to approximate symmetries we describe in Sec.~\ref{sec:measuresymm} below. %

Here we use $\pmin = \pmin(C,B,\epsilon)$ defined as ``the smallest $p$ with which QAOA restricted to linear schedules achieves an approximation ratio of $r=(1-\epsilon)$ for a given instance'' as the metric for quantifying QAOA performance, with the approximation ratio $r$ defined in Eq.~(\ref{eq:qaoa_approx_r}). 
Thus defined, $\pmin$ exists for any problem instance, since from the adiabatic limit QAOA with $p\rightarrow\infty$ can solve any optimization problem exactly~\cite{farhi2014quantum}, and this remains true in our chosen setting of linear schedules. As the \revtwo{time to execute the} QAOA \revtwo{circuit} grows linearly  
with its depth~$p$, and since in the NISQ era the achievable depth~$p$ is limited by the error rates and coherence time of the quantum hardware, $\pmin$ gives a meaningful measure of QAOA performance for our purposes. Note that  here we approach QAOA as a \emph{heuristic} applied to a particular problem \emph{instance} 
or set of instances 
and not 
necessarily 
as an approximation algorithm with performance guarantees on a given problem class.

\subsection{Smooth schedules: a practical approach to QAOA}\label{sec:smoothschedules}

One challenge in both the analysis and the implementation of QAOA is 
the need to %
select sufficiently good algorithm parameters $\vec{\beta},\vec{\gamma}$. Optimizing  parameters for QAOA (e.g., variationally) appear to be difficult in practice~\cite{Shaydulin2019MultistartDOI,Shaydulin2019EvaluatingDOI,khairy2019learning}, presenting two challenges in particular  for numerical experiments. First, performing such numerical studies
requires solving the parameter optimization problem for a large number of instances, which is especially difficult in practice as the number of parameters becomes significant.
Second, even if such data were collected, it would likely not be representative of real-world QAOA performance, since finding optimal parameters for every instance may not be a practical way of using QAOA as the depth~$p$ becomes large.

\begin{wrapfigure}{r}{0.31\textwidth}
  \begin{center}
    \includegraphics[width=0.3\textwidth]{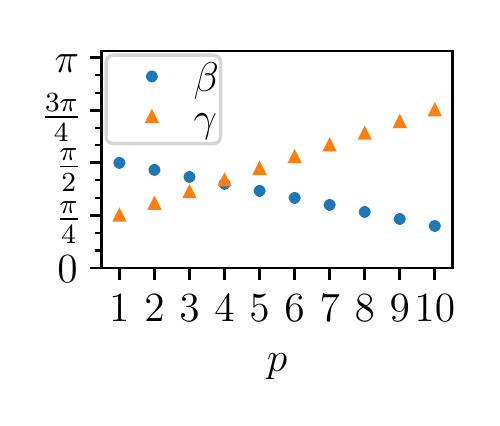}
  \end{center}
  \caption{Example of a linear schedule for $p=10$. Here $\beta_j$ changes linearly from $\pi/2$ at the first QAOA step to $\pi/5$ at the last QAOA step, and $\gamma_j$ changes from $\pi/4$ to $3\pi/4$.}
  \label{fig:linearexample}
\end{wrapfigure}

Instead, more practical approaches to employing QAOA have been proposed. %
One approach, explored in~\cite{khairy2019learning,wilson2019optimizing,verdon2019learning}, among others, involves choosing a class of instances and training a machine learning model to quickly produce high-quality
QAOA parameters for that class. While powerful, however, this approach  resolves only the second challenge, because in order to train the model, one still has to find near-optimal parameters for a sufficiently large dataset representative of the chosen class of instances. 
Another approach, which is used in this %
section, is to specify a class of \emph{parameter schedules} and then fit a schedule of this class to a particular problem instance, such that the number of free parameters is significantly decreased, and 
the difficulty of the parameter optimization problem significantly reduced. 
There are various reasonable choices for the class of schedules; in this work, we focus on \emph{smooth} schedules, that is, schedules where the change from one component to the next in $\vec{\beta}\in \mathbb{R}^p$ and $\vec{\gamma}\in \mathbb{R}^p$ is small relative to the difference between minimum and maximum values of $\beta$ and $\gamma$. Similar approaches to parameter setting for QAOA have been considered %
in~\cite{zhou2018qaoaperformance,crooks2018performance,mnebg2019quantumannealin}, among others. Our particular method is inspired by the reparameterization developed in \cite{zhou2018qaoaperformance}. 
Specifically, we focus on \emph{linear schedules}, that is, schedules where for a given depth~$p$ the component parameters $\beta_j, \gamma_j$ %
change linearly with each step $j$:  
\begin{equation}
\beta_j = a_\beta j+b_\beta, \;\;\;\;\; \gamma_j = a_\gamma j+b_\gamma,
\end{equation}
for $j=1,2,\dots,p$, where $a_\beta, b_\beta, a_\gamma, b_\gamma$ are the remaining free parameters to be optimized. %
An example of such a parameter schedule is presented in Fig.~\ref{fig:linearexample}.

 Restricting QAOA parameters to a particular class of schedules clearly yields a tradeoff between the difficulty of parameter optimization and the algorithm's performance, relative to unrestricted parameters.
 Fitting a particular class of schedules is much easier than finding an optimal general schedule. For example, in general for $p=10$ finding optimal parameters requires optimizing $2\times p = 20$ parameters, whereas fitting a linear schedule requires  optimizing only over 4 parameters (slope and intercept for both $\beta$ and $\gamma$). We also observe this in practice, with multistart methods~\cite{Shaydulin2019MultistartDOI} capable of finding optimal linear schedules quickly and reliably. At the same time, this restriction may increase the minimum depth needed to achieve the desired approximation ratio (i.e., $p_{\min}(C,B,\epsilon) > p_{\min}(C,B,\epsilon,\text{linear schedule})$)\revtwo{, leading to the full power of QAOA with general parameters not being accurately represented by the performance of the restricted variant.} Note that in adiabatic limit $p\rightarrow\infty$ linear schedules are sufficient, so such minimum depth $p_{\min}$ always exists. For the remainder of the section %
 we will  consider QAOA with linear schedules only, i.e. 
$p_{\min} := p_{\min}(C,B,\epsilon,\text{linear schedule})$.

\begin{figure}[b] %
\subfloat[$\pmin$ as a function of the number of nodes. $\pmin$ grows the fastest for the least symmetric graphs (``Trivial'') and the slowest for the most symmetric (complete and star graphs).]{\includegraphics[width=0.7\textwidth]{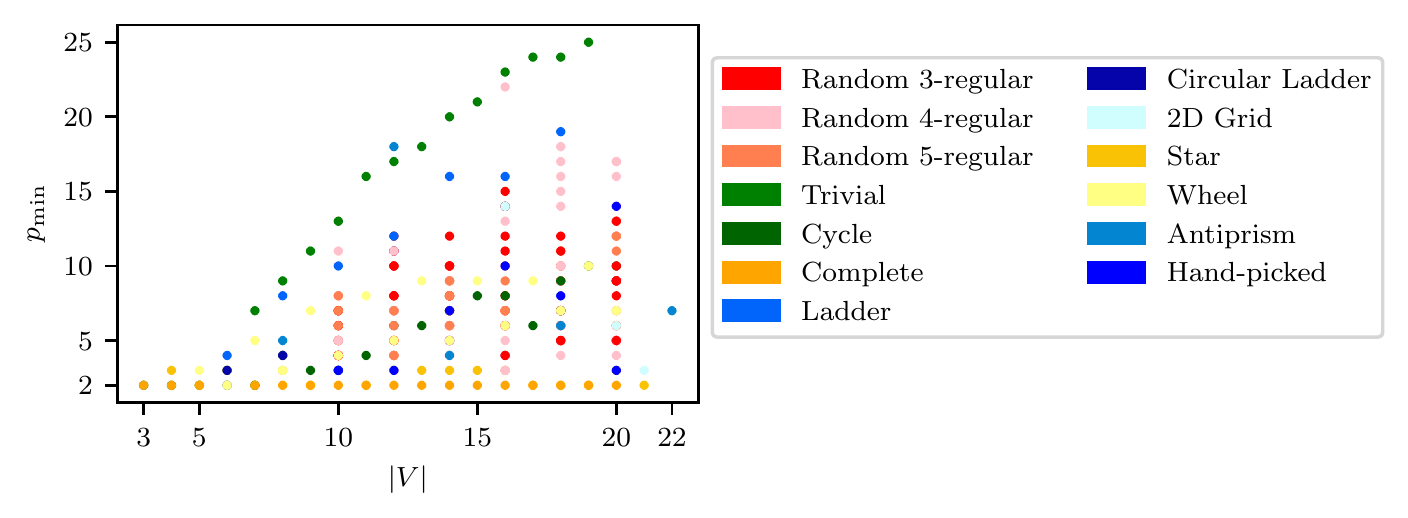}%
\label{fig:pminvsnnodes}}
\hfill
\subfloat[True $\pmin$ and $\pmin$ predicted by support vector regression for both training and testing data. Median absolute error on the test set is $1.37$.]{\includegraphics[width=0.26\textwidth, trim=0.13cm 0cm 0cm 0cm]{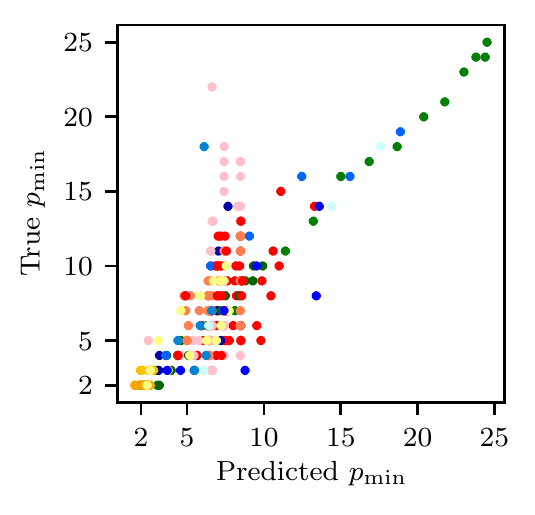}%
\label{fig:svrprediction}}
\caption{Using problem symmetry to predict $\pmin$.}
\label{fig:svrandlegend}
\end{figure}

\begin{table}[b]
    \centering
    \begin{tabular}{l|c|c}
        Feature & Pearson $r$ & p-value \\
        \hline
$\log{Aut(G)}$                 & -0.36397 & 0.00000 \\
$\frac{1}{h'}\sum_{i=1}^{h'}\log{Aut(G'_{i})}$           & -0.32851 & 0.00000 \\
$\frac{1}{h''}\sum_{i,j}\log{Aut(G''_{i,j})}$           & -0.30854 & 0.00000 \\
$|V|$                    & +0.37075 & 0.00000 \\
$|O(G)|$                   & +0.39314 & 0.00000 \\
$\frac{1}{h'}\sum_{i=1}^{h'}|O(G_i)|$             & +0.44326 & 0.00000 \\
$\frac{1}{h''}\sum_{i,j}|O(G''_{i,j})|$             & +0.43035 & 0.00000 \\
$I(G)$           & -0.32295 & 0.00000 \\
$\frac{1}{h'}\sum_{i=1}^{h'}I(G_i)$     & -0.36913 & 0.00000 \\
$\frac{1}{h''}\sum_{i,j}I(G''_{i,j})$     & -0.34254 & 0.00000 
    \end{tabular}
    \caption{Pearson correlation coefficient $r$ and p-value for the test of non-correlation of features with $\pmin$. See Sec.~\ref{sec:measuresymm} for the definitions of the features.}
    \label{tab:correlationfeaturepmin}
\end{table}

\subsection{Exact and approximate symmetry measures}
\label{sec:measuresymm}

Quantitative analysis of graph symmetry is a challenging problem. Over the years, multiple symmetry indices have been introduced to attempt to tackle this problem in different contexts. Some notable examples are graph entropy~\cite{Mowshowitz1968,simonyi1995graph}, index of symmetry~\cite{mowshowitz2010symmetry}, network redundancy~\cite{macarthur2008symmetry}, and normalized network redundancy~\cite{howsymmetricarerealworldgraphs}. 
However, we observe that all of these metrics fail to capture the finer differences between instances required for achieving high predictive power.
Therefore we propose constructing a data-driven symmetry index that is tailored to our particular application. We construct this index by noting that all of the metrics in \cite{Mowshowitz1968,simonyi1995graph,mowshowitz2010symmetry,macarthur2008symmetry,howsymmetricarerealworldgraphs}
are a combination of the number of vertex orbits, the size of orbits, the size of group of automorphisms, the number of nodes of a graph, and the graph entropy. Instead of combining them in an index constructed from some prior intuition, %
we use them as features in a machine learning model. The resulting trained model serves as a symmetry measure of a graph. 

First, we use the following \emph{exact} symmetry measures:

\begin{enumerate}
    \item Logarithm of the size of the group of automorphisms of the graph: $\log{|Aut(G)|}$
    \item Number of vertex orbits of the graph: $|O(G)|$
    \item Entropy of the graph: $I(G) = \frac{1}{n}\sum_i |A_i|\log{|A_i|}$, where $|A_i|$ is the size of $i$th vertex orbit and $n$ is the number of nodes in $G$
\end{enumerate}

Second, we consider the following \emph{approximate} symmetry measures to further improve the performance of our model. Let $G'_{e}$ denote the graph constructed from $G$ by removing edge~$e$. Intuitively, if $G'_{e}$ has an automorphism $\sigma$, then $\sigma$ is said to be an \textit{approximate} automorphism of~$G$. Therefore for each measure of 
exact symmetry %
we can introduce a corresponding measure of \textit{approximate symmetry} of~$G$ defined as the average value of %
the exact measure on $G'_{e}$ when averaged over all edges~$e$. Similarly, we consider $G''_{e_1,e_2}$, %
constructed by removing the two edges $e_1,e_2$ from $G$, and use the symmetry measure of $G''_{e_1,e_2}$ averaged across all pairs of edges ($e_1,e_2$). 
For example, the first metric, 
the logarithm of order of the graph automorphism group, 
induces the following two approximate metrics: the average size of a group of automorphisms for graphs constructed from $G$ by removing one edge ($\frac{1}{h'}\sum_{i=1}^{h'}\log{|Aut(G'_{i})|}$, $h' = {|E| \choose 1} = |E|$), or two edges ($\frac{1}{h''}\sum_{i,j}\log{|Aut(G''_{i,j})|}$, $h'' = {|E| \choose 2}$). Hence, approximate symmetry %
measures %
add another six metrics to the list of features, and in addition we include the number of vertices explicitly, %
for a total of ten features.

Note that if a generating set of the group of graph automorphisms is given, the number of vertex orbits of the graph and their sizes can be computed in polynomial time~\cite{arvindalgebraandcomputation}. This means that the proposed features can be computed in polynomial time from the generating set of the group of automorphisms of the graph and, for approximate features, the induced graphs. 
Recalling the discussion of the complexity of obtaining such generating sets
in Sec.~\ref{sec:measuresymmhardness},  
here we simply reiterate that for many practically interesting problem classes (such as graphs with bounded degree) the full graph symmetry groups and complete set of proposed features can be computed efficiently. 
Additionally, useful bounds on the magnitudes of these symmetry measures are often available; 
in particular, upper bounds for the number of graph automorphisms %
are shown in terms of the vertex degrees (e.g., maximum and average degree of the graph and its  subgraphs) in~\cite{krasikova2002upper}, %
and in terms of the the graph diameter and number of nodes in~\cite{dankelmann2012automorphism}.

\afterpage{%
\begin{antilandscape}
    \centering
    {\footnotesize
    \hspace{-1cm}
    \begin{longtable}{l|c|c|c|c|c|c|c|c|c|c|c|c|c|c|c|c|c|c|c|c|c|c|c}
    Name & 
    \rotatebox{90}{\# graphs} & 
    \rotatebox{90}{$\min |V|$} &  
    \rotatebox{90}{$\max |V|$} &  
    \rotatebox{90}{$\min |E|$} &  
    \rotatebox{90}{$\max |E|$} & 
    \rotatebox{90}{$\min \log{|Aut(G)|}$} &
    \rotatebox{90}{$\max \log{|Aut(G)|}$} & 
    \rotatebox{90}{$\min \frac{1}{h'}\sum_{i=1}^{h'}\log{|Aut(G'_{i})|}$} &
    \rotatebox{90}{$\max \frac{1}{h'}\sum_{i=1}^{h'}\log{|Aut(G'_{i})|}$} &
    \rotatebox{90}{$\min \frac{1}{h''}\sum_{i,j}\log{|Aut(G''_{i,j})|}$} &
    \rotatebox{90}{$\max \frac{1}{h''}\sum_{i,j}\log{|Aut(G''_{i,j})|}$} &
    \rotatebox{90}{$\min |O(G)|$} &
    \rotatebox{90}{$\max |O(G)|$} &
    \rotatebox{90}{$\min \frac{1}{h'}\sum_{i=1}^{h'}|O(G'_i)|$} &
    \rotatebox{90}{$\max \frac{1}{h'}\sum_{i=1}^{h'}|O(G'_i)|$} &
    \rotatebox{90}{$\min \frac{1}{h''}\sum_{i,j}|O(G''_{i,j})|$} &
    \rotatebox{90}{$\max \frac{1}{h''}\sum_{i,j}|O(G''_{i,j})|$} &
    \rotatebox{90}{$\min |I(G)|$} &
    \rotatebox{90}{$\max |I(G)|$} &
    \rotatebox{90}{$\min \frac{1}{h'}\sum_{i=1}^{h'}I(G'_i)$} &
    \rotatebox{90}{$\max \frac{1}{h'}\sum_{i=1}^{h'}I(G'_i)$} &
    \rotatebox{90}{$\frac{1}{h''}\sum_{i,j}I(G''_{i,j})$} &
    \rotatebox{90}{$\frac{1}{h''}\sum_{i,j}I(G''_{i,j})$} \\
\hline
Antiprism & 7 & 6 & 22 & 12 & 44 & 2.8 & 3.9 & 0.7 & 1.4 & 0.4 & 1.0 & 1.0 & 1.0 & 3.0 & 11.5 & 3.5 & 16.7 & 1.8 & 3.1 & 0.6 & 0.7 & 0.3 & 0.6 \\
Circular Ladder & 7 & 6 & 18 & 9 & 27 & 2.5 & 3.9 & 0.9 & 1.4 & 0.5 & 1.0 & 1.0 & 1.0 & 3.0 & 8.3 & 4.2 & 12.4 & 1.8 & 2.9 & 0.7 & 1.0 & 0.5 & 0.8 \\
Complete & 18 & 3 & 20 & 3 & 190 & 1.8 & 42.3 & 0.7 & 37.1 & 0.7 & 33.0 & 1.0 & 1.0 & 2.0 & 2.0 & 2.0 & 2.7 & 1.1 & 3.0 & 0.5 & 2.7 & 0.5 & 2.5 \\
Cycle & 18 & 3 & 20 & 3 & 20 & 1.8 & 3.7 & 0.7 & 0.7 & 0.7 & 1.3 & 1.0 & 1.0 & 2.0 & 10.0 & 2.0 & 10.3 & 1.1 & 3.0 & 0.5 & 0.7 & 0.5 & 0.7 \\
2D Grid & 4 & 16 & 21 & 24 & 42 & 1.4 & 4.4 & 0.2 & 1.4 & 0.1 & 0.6 & 1.0 & 6.0 & 7.5 & 16.0 & 13.5 & 17.0 & 1.2 & 3.0 & 0.2 & 1.1 & 0.1 & 0.5 \\
Hand-picked & 7 & 10 & 20 & 15 & 30 & 4.6 & 5.8 & 1.4 & 2.8 & 0.7 & 1.5 & 1.0 & 1.0 & 3.0 & 7.0 & 5.2 & 10.9 & 2.3 & 3.0 & 0.9 & 1.7 & 0.6 & 1.1 \\
Ladder & 7 & 6 & 18 & 7 & 25 & 1.4 & 1.4 & 0.4 & 0.9 & 0.4 & 1.0 & 2.0 & 5.0 & 3.9 & 14.4 & 3.5 & 15.8 & 1.2 & 1.4 & 0.3 & 0.6 & 0.2 & 0.6 \\
Random 3-reg & 60 & 10 & 20 & 15 & 30 & 0.0 & 3.5 & 0.0 & 1.8 & 0.0 & 1.4 & 3.0 & 20.0 & 7.3 & 20.0 & 7.9 & 20.0 & 0.0 & 1.4 & 0.0 & 0.5 & 0.0 & 0.4 \\
Random 4-reg & 59 & 10 & 20 & 20 & 40 & 0.0 & 2.5 & 0.0 & 1.3 & 0.0 & 0.8 & 4.0 & 20.0 & 7.1 & 20.0 & 8.3 & 20.0 & 0.0 & 1.0 & 0.0 & 0.4 & 0.0 & 0.2 \\
Random 5-reg & 60 & 10 & 20 & 25 & 50 & 0.0 & 1.4 & 0.0 & 0.7 & 0.0 & 0.4 & 4.0 & 20.0 & 8.6 & 20.0 & 9.2 & 20.0 & 0.0 & 1.0 & 0.0 & 0.2 & 0.0 & 0.1 \\
Star & 18 & 4 & 21 & 3 & 20 & 1.8 & 42.3 & 0.7 & 39.3 & 1.4 & 37.1 & 2.0 & 2.0 & 3.0 & 3.0 & 2.0 & 3.0 & 0.8 & 2.9 & 0.3 & 2.7 & 0.6 & 2.5 \\
Trivial & 13 & 7 & 19 & 6 & 18 & 0.0 & 0.0 & 0.8 & 1.1 & 1.6 & 1.7 & 7.0 & 19.0 & 4.7 & 14.6 & 3.9 & 12.7 & 0.0 & 0.0 & 0.3 & 0.5 & 0.5 & 0.7 \\
Wheel & 16 & 5 & 20 & 8 & 38 & 2.1 & 3.6 & 0.7 & 1.0 & 0.5 & 0.9 & 2.0 & 2.0 & 3.0 & 11.0 & 3.3 & 15.4 & 1.1 & 2.8 & 0.5 & 0.6 & 0.3 & 0.5 \\
\hline
Total & 294 & 3 & 22 & 3 & 190  & 0.0 & 42.3 & 0.0 & 39.3 & 0.0 & 37.1 & 1.0 & 20.0 & 2.0 & 20.0 & 2.0 & 20.0 & 0.0 & 3.1 & 0.0 & 2.7 & 0.0 & 2.5 \\
\caption{Description of the dataset. "Trivial" is graph with a trivial group of automorphisms. "Hand-picked" includes various textbook graphs with large groups of automorphisms, e.g. Peterson and Heawood graphs. We make the full dataset available online~\cite{rawdata}.}
    \label{tab:dataset}
    \end{longtable}}

\end{antilandscape}
}

\subsection{Learning the relationship between symmetry and performance}\label{sec:mlsymmetryperfromance}

To numerically study the relationship between the symmetry of the objective function and QAOA performance, we use the \maxcut{} problem on a dataset of 294 graphs, selected as described in Table~\ref{tab:dataset}. We set the target approximation ratio $r=0.95$, with $r$ defined by Eq.~\ref{eq:qaoa_approx_r}.
Therefore the metric to be predicted is given by $p_{\min} = $ ``the smallest $p$ with which QAOA restricted to linear schedules achieves an approximation ratio of $0.95$ for \maxcut{} on a given instance.'' %
While this choice of target approximation ratio is somewhat arbitrary, it is inspired by %
classical hardness results for MaxCut, which is NP-hard to approximate better than $\frac{16}{17}\approx 0.941$ in the worst-case~\cite{Hstad2001}. 
Under the Unique Games conjecture, a weaker assumption than P$\neq$NP, the~$0.87856$ approximation ratio achieved by the Goemans--Williamson algorithm is optimal~\cite{goemans1995improved,khot2007optimal}. Indeed, for similar complexity reasons we do not expect quantum computers to efficiently solve NP-hard problems such as MaxCut. 
In any case, these complexity results become meaningful only as problem sizes becomes large, which %
makes %
comparison with results based on fixed size training and prediction sets difficult. %
Hence, we select an approximation ratio of~$0.95$ as a reasonable target that reflects %
desirable QAOA performance and the possibility of quantum advantage %
beyond the finite dataset of this study.

We use \texttt{nauty}~\cite{McKay201494} to compute the features described in Sec.~\ref{sec:measuresymm}. We compute $\pmin$ for each problem in the dataset by considering iteratively larger depths $p$, starting with $p=2$ and optimizing the linear parameter schedule as defined in Sec.~\ref{sec:smoothschedules}.
until %
the target approximation ratio is achieved. We use COBYLA~\cite{powell1994direct,powell1998direct} implemented in the SciPy~\cite{scipy} package as a local optimizer in the libEnsemble~\cite{libEnsemble_0.5.0} implementation of APOSMM~\cite{LarWild14,LW16}. We use NetworkX~\cite{hagberg2008-rs} for graph operations and GNU-Parallel for large-scale experiments~\cite{tange_ole_2018_1146014}. The code is available online at~\cite{code}.

\subsubsection{Initial observations from the training set}

We begin by discussing some notable properties of the training set. First, we observe 
that all the chosen features correlate with $\pmin$, with %
Pearson correlation coefficients between 0.3 and 0.5 in absolute value, 
where the Pearson correlation coefficient is defined as $\frac{\sum_{i=1}^n(x_i-\bar{x})(y_i-\bar{x})}{\sqrt{\sum_{i=1}^n(x_i-\bar{x})^2}\sqrt{\sum_{i=1}^n(y_i-\bar{y})^2}}$, and  $\bar{x}=\frac{1}{n}\sum_{i=1}^nx_i$ is the sample mean. The empirically observed correlation coefficients are presented in Table~\ref{tab:correlationfeaturepmin}. 
The large value of the correlation coefficients
indicates that the chosen features are meaningful and have predictive power. We note that measures of approximate symmetry %
have correlation coefficients similar to the corresponding exact symmetries, as expected.%

Second, we observe that the hardest problem instances for QAOA (with respect to $\pmin$) in our dataset are found to be the ones  corresponding to the least symmetric graphs, and the easiest are the ones corresponding to the most symmetric graphs, as one may expect from our results of Sec.~\ref{sec:theory}. Fig.~\ref{fig:pminvsnnodes} shows the scaling of $\pmin$ with problem size, with the hardest instances seen to be those %
for which $\pmin$ grows relatively quickly, %
and the easiest %
those for which $\pmin$ grows slowly or does not grow at all. %
For problem instances with almost no symmetry (graphs with trivial automorphism group), $\pmin$ is observed to grow the fastest, and for instances with the most symmetry (e.g., complete graphs) %
$\pmin$ grows the slowest. These results %
provide further evidence of the connection between symmetry and QAOA performance, though we emphasize that more work is needed to better understand to what degree these observations %
hold asymptotically, or, more generally, in the setting of real-world quantum devices. 

\subsubsection{Machine learning approaches to predicting QAOA performance}

To quantify the intuition arising from Fig.~\ref{fig:pminvsnnodes}, we use the metrics discussed in Sec.~\ref{sec:measuresymm} as features for a machine learning model that we train to predict $\pmin$. We reserve 30\% on the dataset as the testing set and use the rest for training the model, with the training set further partitioned for cross-validation as specified below. We approach the problem in two ways. In the first approach, we treat the task of predicting $\pmin$ as a regression problem. In the second approach, we use an ensemble of classifiers to predict $\pmin$. Below we discuss both approaches and the tradeoffs they introduce. We choose median absolute error as our target metric. This choice is motivated by the physical meaning of $\pmin$, namely, the minimum depth required to achieve 0.95 approximation ratio by using QAOA with linear schedules. Since $\pmin$ is an integer, absolute error is an easily interpretable metric. Moreover, an absolute error of one or close to one is tolerable for most applications (e.g., if we want to use the trained model to predict the value of depth $p$ to use in QAOA or if we want to establish whether a given problem requires depth beyond the capabilities of target quantum hardware).

First, we approach the problem of predicting QAOA performance (i.e., predicting the number $\pmin$ for a given problem instance) directly as a regression problem. We use support vector regression (SVR)~\cite{drucker1997support,scikit-learn,chang2011libsvm} with the radial basis function kernel. We use 5-fold cross-validation stratified by graph class for hyperparameter optimization. We achieve 0.73 median absolute error on the training set and 1.37 on the testing set (i.e., on the instances previously unseen by the model). We observe a correlation between the predicted $\pmin$ and true $\pmin$ with a Pearson correlation coefficient of 0.71 on the test set. The results are visualized in Fig.~\ref{fig:svrprediction}.

The second approach we use for predicting $\pmin$ is training an ensemble of classifiers. We simulate a realistic scenario of limited circuit depth by grouping all instances with $\pmin\geq 15$ into one class corresponding to ``depth beyond the capabilities of the target hardware.'' Since $\pmin$ is a discrete value (an integer), classification is a natural choice. An issue with using a classifier directly is that assigning each value of $\pmin$ to be a class discards the information that two consecutive integers are more similar than two integers that are far apart. To address this issue, we instead train an ensemble of binary classifiers, where each classifier answers the question ``Is the value of $\pmin$ smaller than the specified cutoff?'' Using an ensemble of ``cutoff'' classifiers is a standard approach for ordinal regression~\cite{waegeman2009ensemble,perez2013projection}. Each classifier is a support vector machine classifier~\cite{scikit-learn,chang2011libsvm}. We use the radial basis function kernel and the optimal hyperparameters found by cross-validation for the support vector regression. 

\begin{figure}[htbp]
\subfloat[The blue dots are distance to the decision boundary for each classifier. The blue line is the least squares quadratic fit. The yellow dot is the predicted $\pmin$, and the red vertical line is the true $\pmin$.]{\includegraphics[width=0.34\textwidth, trim=0.3cm 0cm 0cm 0cm]{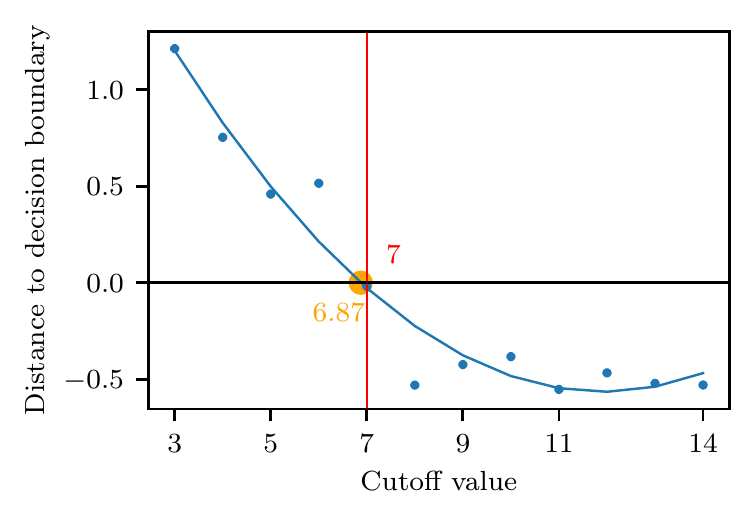}%
\label{fig:svcquadfit1}}
\hfill
\subfloat[If the least squares quadratic fit does not intersect the x-axis within the bounds, we choose either the minimum or maximum $\pmin$ based on the majority vote of classifiers. Here all classifiers put the point in the first class (``$\pmin$ is smaller than the cutoff'').]{\includegraphics[width=0.34\textwidth, trim=0.3cm 0cm 0cm 0cm]{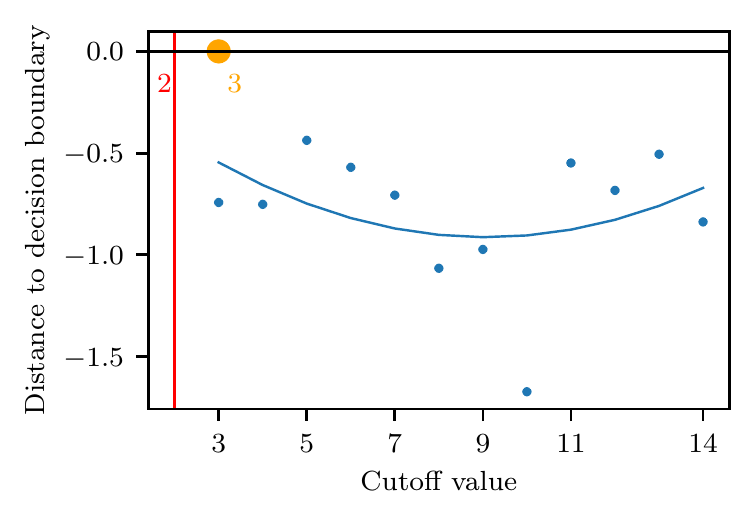}%
\label{fig:svcquadfit2}}
\hfill
\subfloat[True $\pmin$ and $\pmin$ predicted by an ensemble of support vector machines for both training and testing data. Median absolute error on the test set is $1.35$. See Fig.~\ref{fig:svrandlegend} for the color legend.]{\includegraphics[width=0.26\textwidth]{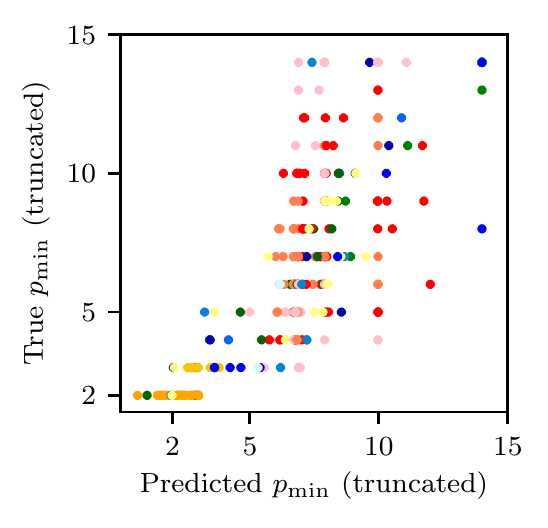}%
\label{fig:svcprediction}}
\caption{Using an ensemble of SVM classifiers to predict $\pmin$.}
\label{fig:svcall}
\end{figure}

The ensemble of ``cutoff'' binary classifiers is used to predict $\pmin$ in the following way. For a new point, we compute the distance to the decision boundary for each classifier. By convention, if the distance is positive, then the point belongs to the first class ($\pmin$ is smaller than the cutoff); otherwise it belongs to the second class ($\pmin$ is greater than or equal to the cutoff). Therefore the minimum cutoff value at which the binary classifiers are assigning the point to the second class (i.e., at which the distance to the decision boundary changes sign) is the $\pmin$ for this point. To make the distances to the decision boundary directly comparable for all classifiers, we standardize them in the following way. For each classifier we compute distances to the decision boundary on the entire training set. If the standard deviation of these distances is $\sigma$, we let the standardized distance to the decision boundary be $d_z=d/\sigma$, where $d$ is the distance to the decision boundary for the new point. To predict $\pmin$ for the point, we fit a least squares quadratic function to the standardized distances to the decision boundary as a function of the cutoff value for all classifiers, and we use its intersection with the x-axis as the predicted value. This process is visualized in Figures~\ref{fig:svcquadfit1} and \ref{fig:svcquadfit2}. We achieve a median absolute error of 1.35 on the test set. 

The two approaches considered for predicting $\pmin$ demonstrate similarly good
performance, as %
observed by 
comparing Figures~\ref{fig:svrprediction} and \ref{fig:svcprediction}. As predicted by the high correlation between an exact symmetry feature and the corresponding approximate symmetry features, we observe only a small decrease in performance ($\approx 10\%$ increase in test error) if we discard the approximate symmetry features and use only the three exact ones.
From the observed high ratio of support vectors, we expect that performance can be significantly improved by introducing more datapoints (problem instances). Whereas in this study the size of the dataset was limited by the cost of classical simulation, recent advances in QAOA simulation~\cite{huang2019alibabacloud,yuriQAOAsim} and the increasing availability of quantum hardware are widely believed to enable large-scale studies of QAOA that we expect to significantly boost the performance of the proposed methods.

\section{Discussion and future work} \label{sec:discussion}

In this paper we introduced a formal connection between the dynamical symmetries of QAOA and the classical symmetries of the underlying objective function, which we employed %
to derive several operational results. %
We further demonstrated the power of this connection by training machine learning models to predict QAOA performance %
solely from the information about exact and approximate symmetries of the problem. As emphasized, our approach is general 
and applies to a wide variety of optimization problems, 
and may be easily extended to %
generalizations of QAOA beyond the transverse-field mixer and uniform superposition initial state. 

An especially attractive future  research direction is to apply our %
results towards deriving improved %
bounds concerning QAOA performance and conditions for 
(in)efficient classical simulation (i.e., to better characterize regimes of potential quantum advantage),
in the spirit of the %
approaches of \cite{braviyobstacles} and \cite{farhi2016quantum}.  Particularly desirable is the further identification of classes of problems where symmetry leads to provable performance advantages, or limitations. 
In any case, we are optimistic that symmetry  
insights 
may be an important tool in the design of better quantum ansatz and algorithms for optimization, and beyond.  

Another interesting future direction is to  %
consider further %
generalizations of our notion of symmetry. 
Since in some cases QAOA may not \lq\lq see the whole graph\rq\rq%
~\cite{farhi2020qaoaneedsfullgraph}, symmetries involving particular subsets of variables (subgraphs) as dictated by the structure of the cost function may be even more precise predictors of QAOA performance. Research in these directions is complicated by the fact that such symmetries may  come into play only for much larger graphs, as well as other combinatoric considerations limiting 
our ability to perform numerical studies. 
Indeed, a limitation to applying our machine learning approach in practice is that for some classes of problems the metrics defined in Sec.~\ref{sec:measuresymm} may not be computable in polynomial time in the worst case. %
While for problem instances with %
hundreds or thousands of variables off-the-shelf tools like \texttt{nauty}~\cite{McKay201494} can compute them in seconds, this approach is not scalable to graphs with hundreds of thousands or millions of nodes.
Therefore an important next step is 
exploring methods to efficiently approximate these quantities with techniques such as network alignment~\cite{qiu2019elruna}.  
Interpretability of the machine learning approach may also evolve with both the size of instances and consequently the training set size. Although the SVM/SVR approach is one of the most interpretable learning models that helps  estimate the importance of features and their combinations, advanced scalable nonlinear methods such as those presented in~\cite{sadrfaridpour2019engineering} will likely be required in order to avoid overfitting of the model.

\section*{Acknowledgments}

We thank Salvatore Mandr\`a, Gianni Mossi, Stefan Wild, Ehsan Sadrfaridpour, Justin Sybrandt, Jeffrey Larson, and  Mikhail Klin for valuable discussions and assistance in the preparation of this manuscript. 
  R.S. is grateful to KBR, the QuAIL group, and the NASA Ames Research Center for the internship opportunity that enabled this research. R.S. and I.S. were supported by funding from the Defense Advanced Research Projects Agency (DARPA). 
 S.H., T.H, and R.S appreciate support from the NASA Ames Research Center and support from DARPA under IAA 8839, Annex 114.
 S.H. and T.H. were supported in part by the NASA Academic Mission Services under Contract NNA16BD14C. 
R.S. was supported by Laboratory Directed Research and Development (LDRD) funding from Argonne National Laboratory, provided by the Director, Office of Science, of the U.S. Department of Energy under Contract No. DE-AC02-06CH11357 and by the U.S. Department of Energy, Office of Science, Office of Advanced Scientific Computing Research, Accelerated Research for Quantum Computing program. 
The views, opinions and/or findings expressed are those of the author and should not be interpreted as representing the official views or policies of the Department of Defense or the U.S. Government.

\bibliographystyle{plainnat}
\bibliography{references}

\begin{thebibliography}{87}
\providecommand{\natexlab}[1]{#1}
\providecommand{\url}[1]{\texttt{#1}}
\expandafter\ifx\csname urlstyle\endcsname\relax
  \providecommand{\doi}[1]{doi: #1}\else
  \providecommand{\doi}{doi: \begingroup \urlstyle{rm}\Url}\fi

\bibitem[cod()]{code}
URL \url{https://github.com/rsln-s/Classical-symmetries-and-QAOA}.
\newblock [Online.].

\bibitem[raw()]{rawdata}
URL \url{https://www.dropbox.com/s/cftspvdoovnzi4l/allresults.p.zip?dl=0}.
\newblock [Online.].

\bibitem[Arora et~al.(1999)Arora, Karger, and Karpinski]{Arora1999}
Sanjeev Arora, David Karger, and Marek Karpinski.
\newblock Polynomial time approximation schemes for dense instances of
  {NP}-hard problems.
\newblock \emph{Journal of Computer and System Sciences}, 58\penalty0
  (1):\penalty0 193--210, February 1999.

\bibitem[Arute et~al.(2019)]{Arute2019}
Frank Arute et~al.
\newblock Quantum supremacy using a programmable superconducting processor.
\newblock \emph{Nature}, 574\penalty0 (7779):\penalty0 505--510, October 2019.

\bibitem[Arute et~al.(2020)]{googleqaoaexperimental}
Frank Arute et~al.
\newblock Quantum approximate optimization of non-planar graph problems on a
  planar superconducting processor.
\newblock \emph{arXiv:2004.04197}, 2020.

\bibitem[Arving(2007)]{arvindalgebraandcomputation}
V.~Arving.
\newblock Lecture notes, 2007.
\newblock URL
  \url{https://www.cmi.ac.in/~ramprasad/lecturenotes/algcomp/tillnow.pdf}.
\newblock [Online.].

\bibitem[Babai(2015)]{babai2015giquasipoly}
László Babai.
\newblock Graph isomorphism in quasipolynomial time.
\newblock \emph{arXiv:1512.03547}, 2015.

\bibitem[Balasubramanian(1994)]{Balasubramanian1994}
K.~Balasubramanian.
\newblock Computer generation of automorphism groups of weighted graphs.
\newblock \emph{Journal of Chemical Information and Modeling}, 34\penalty0
  (5):\penalty0 1146--1150, September 1994.

\bibitem[Ball and Geyer-Schulz(2018)]{howsymmetricarerealworldgraphs}
Fabian Ball and Andreas Geyer-Schulz.
\newblock How symmetric are real-world graphs? {A} large-scale study.
\newblock \emph{Symmetry}, 10\penalty0 (1):\penalty0 29, Jan 2018.

\bibitem[Bapat and Jordan(2019)]{bapat2019hammingramp}
Aniruddha Bapat and Stephen Jordan.
\newblock Bang-bang control as a design principle for classical and quantum
  optimization algorithms.
\newblock \emph{Quantum Info. Comput.}, 19\penalty0 (5–6):\penalty0
  424–446, May 2019.

\bibitem[Barkoutsos et~al.(2020)Barkoutsos, Nannicini, Robert, Tavernelli, and
  Woerner]{Barkoutsos2020}
Panagiotis~Kl. Barkoutsos, Giacomo Nannicini, Anton Robert, Ivano Tavernelli,
  and Stefan Woerner.
\newblock Improving variational quantum optimization using {CVaR}.
\newblock \emph{Quantum}, 4:\penalty0 256, April 2020.

\bibitem[Barnhart et~al.(1998)Barnhart, Johnson, Nemhauser, Savelsbergh, and
  Vance]{barnhart1998branch}
Cynthia Barnhart, Ellis~L Johnson, George~L Nemhauser, Martin~WP Savelsbergh,
  and Pamela~H Vance.
\newblock Branch-and-price: Column generation for solving huge integer
  programs.
\newblock \emph{Operations research}, 46\penalty0 (3):\penalty0 316--329, 1998.

\bibitem[Ben-David et~al.(2020)Ben-David, Childs, Gilyén, Kretschmer, Podder,
  and Wang]{2006.12760}
Shalev Ben-David, Andrew~M. Childs, András Gilyén, William Kretschmer,
  Supartha Podder, and Daochen Wang.
\newblock Symmetries, graph properties, and quantum speedups.
\newblock \emph{arXiv:2006.12760}, 2020.

\bibitem[Berman and Karpinski(1999)]{Berman1999}
Piotr Berman and Marek Karpinski.
\newblock On some tighter inapproximability results (extended abstract).
\newblock In \emph{Automata, Languages and Programming}, pages 200--209.
  Springer Berlin Heidelberg, 1999.

\bibitem[Biggs(1974)]{biggs1993algebraic}
Norman Biggs.
\newblock \emph{Algebraic Graph Theory}.
\newblock Cambridge University Press, May 1974.

\bibitem[Bodlaender(1990)]{bodlaender1990polynomial}
Hans~L Bodlaender.
\newblock Polynomial algorithms for graph isomorphism and chromatic index on
  partial k-trees.
\newblock \emph{Journal of Algorithms}, 11\penalty0 (4):\penalty0 631--643,
  1990.

\bibitem[Bravyi et~al.(2019)Bravyi, Kliesch, Koenig, and Tang]{braviyobstacles}
Sergey Bravyi, Alexander Kliesch, Robert Koenig, and Eugene Tang.
\newblock Obstacles to state preparation and variational optimization from
  symmetry protection.
\newblock \emph{arXiv:1910.08980}, 2019.

\bibitem[Bringewatt and Jarret(2020)]{Bringewatt2020}
Jacob Bringewatt and Michael Jarret.
\newblock Effective gaps are not effective: Quasipolynomial classical
  simulation of obstructed stoquastic hamiltonians.
\newblock \emph{Physical Review Letters}, 125\penalty0 (17), October 2020.
\newblock \doi{10.1103/physrevlett.125.170504}.
\newblock URL \url{https://doi.org/10.1103/physrevlett.125.170504}.

\bibitem[Bärtschi and Eidenbenz(2020)]{bartschi2020grover}
Andreas Bärtschi and Stephan Eidenbenz.
\newblock Grover mixers for {QAOA}: Shifting complexity from mixer design to
  state preparation.
\newblock \emph{arXiv:2006.00354}, 2020.

\bibitem[Chang and Lin(2011)]{chang2011libsvm}
Chih-Chung Chang and Chih-Jen Lin.
\newblock {LIBSVM}.
\newblock \emph{{ACM} Transactions on Intelligent Systems and Technology},
  2\penalty0 (3):\penalty0 1--27, April 2011.

\bibitem[Chiu et~al.(2016)Chiu, Teo, Schnyder, and Ryu]{chiu2016classification}
Ching-Kai Chiu, Jeffrey~CY Teo, Andreas~P Schnyder, and Shinsei Ryu.
\newblock Classification of topological quantum matter with symmetries.
\newblock \emph{Reviews of Modern Physics}, 88\penalty0 (3):\penalty0 035005,
  2016.

\bibitem[Crooks(2018)]{crooks2018performance}
Gavin~E Crooks.
\newblock Performance of the quantum approximate optimization algorithm on the
  maximum cut problem.
\newblock \emph{arXiv:1811.08419}, 2018.

\bibitem[Dankelmann et~al.(2012)Dankelmann, Erwin, Mukwembi, Rodrigues,
  Mwambene, and Sabidussi]{dankelmann2012automorphism}
Peter Dankelmann, David Erwin, Simon Mukwembi, Bernardo~Gabriel Rodrigues,
  E~Mwambene, and Gert Sabidussi.
\newblock Automorphism group and diameter of a graph.
\newblock \emph{Journal of Graph Theory}, 70\penalty0 (1):\penalty0 80--91,
  2012.

\bibitem[Darga et~al.(2004)Darga, Liffiton, Sakallah, and
  Markov]{darga2004exploiting}
Paul~T Darga, Mark~H Liffiton, Karem~A Sakallah, and Igor~L Markov.
\newblock Exploiting structure in symmetry detection for {CNF}.
\newblock In \emph{Proceedings of the 41st annual Design Automation
  Conference}, pages 530--534, 2004.

\bibitem[Drucker et~al.(1997)Drucker, Burges, Kaufman, Smola, and
  Vapnik]{drucker1997support}
Harris Drucker, Christopher~JC Burges, Linda Kaufman, Alex~J Smola, and
  Vladimir Vapnik.
\newblock Support vector regression machines.
\newblock In \emph{Advances in neural information processing systems}, pages
  155--161, 1997.

\bibitem[Eldar and Harrow(2017)]{Eldar2017}
Lior Eldar and Aram~W. Harrow.
\newblock Local {H}amiltonians whose ground states are hard to approximate.
\newblock In \emph{2017 {IEEE} 58th Annual Symposium on Foundations of Computer
  Science ({FOCS})}. {IEEE}, October 2017.

\bibitem[Erd{\H{o}}s and R{\'e}nyi(1963)]{erdHos1963asymmetric}
Paul Erd{\H{o}}s and Alfr{\'e}d R{\'e}nyi.
\newblock Asymmetric graphs.
\newblock \emph{Acta Mathematica Academiae Scientiarum Hungarica}, 14\penalty0
  (3-4):\penalty0 295--315, 1963.

\bibitem[Farhi and Harrow(2016)]{farhi2016quantum}
Edward Farhi and Aram~W Harrow.
\newblock Quantum supremacy through the quantum approximate optimization
  algorithm.
\newblock \emph{arXiv:1602.07674}, 2016.

\bibitem[Farhi et~al.(2014{\natexlab{a}})Farhi, Goldstone, and
  Gutmann]{farhi2014quantum}
Edward Farhi, Jeffrey Goldstone, and Sam Gutmann.
\newblock A quantum approximate optimization algorithm.
\newblock \emph{arXiv:1411.4028}, 2014{\natexlab{a}}.

\bibitem[Farhi et~al.(2014{\natexlab{b}})Farhi, Goldstone, and
  Gutmann]{farhi2014quantumbounded}
Edward Farhi, Jeffrey Goldstone, and Sam Gutmann.
\newblock A quantum approximate optimization algorithm applied to a bounded
  occurrence constraint problem.
\newblock \emph{arXiv:1412.6062}, 2014{\natexlab{b}}.

\bibitem[Farhi et~al.(2020)Farhi, Gamarnik, and
  Gutmann]{farhi2020qaoaneedsfullgraph}
Edward Farhi, David Gamarnik, and Sam Gutmann.
\newblock The {Quantum Approximate Optimization Algorithm} needs to see the
  whole graph: A typical case.
\newblock \emph{arXiv:2004.09002}, 2020.

\bibitem[Filotti and Mayer(1980)]{filotti1980polynomial}
Ion~Stefan Filotti and Jack~N Mayer.
\newblock A polynomial-time algorithm for determining the isomorphism of graphs
  of fixed genus.
\newblock In \emph{Proceedings of the twelfth annual ACM symposium on Theory of
  computing}, pages 236--243, 1980.

\bibitem[Goemans and Williamson(1995)]{goemans1995improved}
Michel~X. Goemans and David~P. Williamson.
\newblock Improved approximation algorithms for maximum cut and satisfiability
  problems using semidefinite programming.
\newblock \emph{Journal of the {ACM} ({JACM})}, 42\penalty0 (6):\penalty0
  1115--1145, November 1995.

\bibitem[Grohe(2012{\natexlab{a}})]{grohe2012fixed}
Martin Grohe.
\newblock Fixed-point definability and polynomial time on graphs with excluded
  minors.
\newblock \emph{Journal of the ACM (JACM)}, 59\penalty0 (5):\penalty0 1--64,
  2012{\natexlab{a}}.

\bibitem[Grohe(2012{\natexlab{b}})]{grohe2012structural}
Martin Grohe.
\newblock Structural and logical approaches to the graph isomorphism problem.
\newblock In \emph{SODA}, page 188, 2012{\natexlab{b}}.

\bibitem[Guerreschi and Matsuura(2019)]{Guerreschi2019}
G.~G. Guerreschi and A.~Y. Matsuura.
\newblock {QAOA} for max-cut requires hundreds of qubits for quantum speed-up.
\newblock \emph{Scientific Reports}, 9\penalty0 (1), May 2019.

\bibitem[Hadfield(2018{\natexlab{a}})]{hadfield2018quantum}
Stuart Hadfield.
\newblock Quantum algorithms for scientific computing and approximate
  optimization.
\newblock \emph{Columbia university PhD dissertation, arXiv:1805.03265},
  2018{\natexlab{a}}.

\bibitem[Hadfield(2018{\natexlab{b}})]{hadfieldrepresentation}
Stuart Hadfield.
\newblock On the representation of {Boolean} and real functions as
  {Hamiltonians} for quantum computing.
\newblock \emph{arXiv:1804.09130}, 2018{\natexlab{b}}.

\bibitem[Hadfield et~al.(2019)Hadfield, Wang, O'Gorman, Rieffel, Venturelli,
  and Biswas]{hadfield2017quantum}
Stuart Hadfield, Zhihui Wang, Bryan O'Gorman, Eleanor Rieffel, Davide
  Venturelli, and Rupak Biswas.
\newblock From the {Quantum Approximate Optimization Algorithm} to a quantum
  alternating operator ansatz.
\newblock \emph{Algorithms}, 12\penalty0 (2):\penalty0 34, February 2019.

\bibitem[Hagberg et~al.(2008)Hagberg, Schult, and Swart]{hagberg2008-rs}
Aric~A. Hagberg, Daniel~A. Schult, and Pieter~J. Swart.
\newblock Exploring network structure, dynamics, and function using {NetworkX}.
\newblock In Ga\"el Varoquaux, Travis Vaught, and Jarrod Millman, editors,
  \emph{Proceedings of the 7th Python in Science Conference (SciPy 2008)},
  pages 11--15, Pasadena, CA USA, 2008.

\bibitem[H{\aa}stad(2001)]{Hstad2001}
Johan H{\aa}stad.
\newblock Some optimal inapproximability results.
\newblock \emph{Journal of the {ACM}}, 48\penalty0 (4):\penalty0 798--859, July
  2001.

\bibitem[Hastings(2013)]{hastings2012tlescommutingham}
Matthew~B. Hastings.
\newblock Trivial low energy states for commuting {H}amiltonians, and the
  quantum {PCP} conjecture.
\newblock \emph{Quantum Info. Comput.}, 13\penalty0 (5–6):\penalty0
  393–429, 2013.

\bibitem[Hastings(2019)]{hastings2019}
Matthew~B. Hastings.
\newblock Classical and quantum bounded depth approximation algorithms.
\newblock \emph{arXiv:1905.07047}, 2019.

\bibitem[Huang et~al.(2019)Huang, Szegedy, Zhang, Gao, Chen, and
  Shi]{huang2019alibabacloud}
Cupjin Huang, Mario Szegedy, Fang Zhang, Xun Gao, Jianxin Chen, and Yaoyun Shi.
\newblock Alibaba cloud quantum development platform: Applications to quantum
  algorithm design, 2019.

\bibitem[Hudson et~al.(2019)Hudson, Larson, Wild, and
  Bindel]{libEnsemble_0.5.0}
Stephen Hudson, Jeffrey Larson, Stefan~M. Wild, and David Bindel.
\newblock {libEnsemble} users manual, 2019.
\newblock URL
  \url{https://buildmedia.readthedocs.org/media/pdf/libensemble/latest/libensemble.pdf}.

\bibitem[Jiang et~al.(2017)Jiang, Rieffel, and Wang]{jiang2017near}
Zhang Jiang, Eleanor~G. Rieffel, and Zhihui Wang.
\newblock Near-optimal quantum circuit for {G}rover's unstructured search using
  a transverse field.
\newblock \emph{Physical Review A}, 95\penalty0 (6), June 2017.

\bibitem[Jones et~al.(2001--)Jones, Oliphant, Peterson, et~al.]{scipy}
Eric Jones, Travis Oliphant, Pearu Peterson, et~al.
\newblock {SciPy}: Open source scientific tools for {Python}, 2001--.
\newblock URL \url{http://www.scipy.org/}.
\newblock [Online.].

\bibitem[Junttila and Kaski(2007)]{junttila2007engineering}
Tommi Junttila and Petteri Kaski.
\newblock Engineering an efficient canonical labeling tool for large and sparse
  graphs.
\newblock In \emph{2007 Proceedings of the Ninth Workshop on Algorithm
  Engineering and Experiments (ALENEX)}, pages 135--149. SIAM, 2007.

\bibitem[Khairy et~al.(2019)Khairy, Shaydulin, Cincio, Alexeev, and
  Balaprakash]{khairy2019learning}
Sami Khairy, Ruslan Shaydulin, Lukasz Cincio, Yuri Alexeev, and Prasanna
  Balaprakash.
\newblock Learning to optimize variational quantum circuits to solve
  combinatorial problems.
\newblock \emph{Proceedings of the Thirty-Forth AAAI Conference on Artificial
  Intelligence (AAAI-20)}, 2019.

\bibitem[Khot et~al.(2007)Khot, Kindler, Mossel, and
  O’Donnell]{khot2007optimal}
Subhash Khot, Guy Kindler, Elchanan Mossel, and Ryan O’Donnell.
\newblock {Optimal inapproximability results for MAX-CUT and other 2-variable
  CSPs?}
\newblock \emph{SIAM Journal on Computing}, 37\penalty0 (1):\penalty0 319--357,
  2007.

\bibitem[Kim et~al.(2002)Kim, Sudakov, and Vu]{kim2002asymmetry}
Jeong~Han Kim, Benny Sudakov, and Van~H Vu.
\newblock On the asymmetry of random regular graphs and random graphs.
\newblock \emph{Random Structures \& Algorithms}, 21\penalty0 (3-4):\penalty0
  216--224, 2002.

\bibitem[Krasikov et~al.(2002)Krasikov, Lev, and Thatte]{krasikova2002upper}
Ilia Krasikov, Arieh Lev, and Bhalchandra~D Thatte.
\newblock Upper bounds on the automorphism group of a graph.
\newblock \emph{Discrete Mathematics}, 256:\penalty0 489--493, 2002.

\bibitem[Larson and Wild(2016)]{LarWild14}
Jeffrey Larson and Stefan~M. Wild.
\newblock A batch, derivative-free algorithm for finding multiple local minima.
\newblock \emph{Optimization and Engineering}, 17\penalty0 (1):\penalty0
  205--228, 2016.

\bibitem[Larson and Wild(2018)]{LW16}
Jeffrey Larson and Stefan~M. Wild.
\newblock Asynchronously parallel optimization solver for finding multiple
  minima.
\newblock \emph{Mathematical Programming Computation}, 10\penalty0
  (3):\penalty0 303--332, 2018.

\bibitem[L{\'o}pez-Presa et~al.(2014)L{\'o}pez-Presa, Chiroque, and
  Fern{\'a}ndez~Anta]{lopez2014novel}
Jos{\'e}~Luis L{\'o}pez-Presa, Luis~F Chiroque, and Antonio Fern{\'a}ndez~Anta.
\newblock Novel techniques to speed up the computation of the automorphism
  group of a graph.
\newblock \emph{Journal of Applied Mathematics}, 2014, 2014.

\bibitem[Luks(1982)]{luks1982isomorphism}
Eugene~M Luks.
\newblock Isomorphism of graphs of bounded valence can be tested in polynomial
  time.
\newblock \emph{Journal of computer and system sciences}, 25\penalty0
  (1):\penalty0 42--65, 1982.

\bibitem[Lykov et~al.(2020)Lykov, Schutski, Galda, Vinokur, and
  Alexeev]{yuriQAOAsim}
Danylo Lykov, Roman Schutski, Alexey Galda, Valerii Vinokur, and Yurii Alexeev.
\newblock Tensor network quantum simulator with step-dependent parallelization.
\newblock \emph{arXiv:2012.02430}, 2020.

\bibitem[MacArthur et~al.(2008)MacArthur, S{\'{a}}nchez-Garc{\'{\i}}a, and
  Anderson]{macarthur2008symmetry}
Ben~D. MacArthur, Rub{\'{e}}n~J. S{\'{a}}nchez-Garc{\'{\i}}a, and James~W.
  Anderson.
\newblock Symmetry in complex networks.
\newblock \emph{Discrete Applied Mathematics}, 156\penalty0 (18):\penalty0
  3525--3531, November 2008.

\bibitem[Mbeng et~al.(2019)Mbeng, Fazio, and Santoro]{mnebg2019quantumannealin}
Glen~Bigan Mbeng, Rosario Fazio, and Giuseppe Santoro.
\newblock Quantum annealing: a journey through digitalization, control, and
  hybrid quantum variational schemes.
\newblock \emph{arXiv:1906.08948}, 2019.

\bibitem[McKay(1978)]{mckay1978computing}
Brendan~D McKay.
\newblock Computing automorphisms and canonical labellings of graphs.
\newblock In \emph{Combinatorial mathematics}, pages 223--232. Springer, 1978.

\bibitem[McKay and Piperno(2014)]{McKay201494}
Brendan~D. McKay and Adolfo Piperno.
\newblock Practical graph isomorphism, {II}.
\newblock \emph{Journal of Symbolic Computation}, 60\penalty0 (0):\penalty0 94
  -- 112, 2014.

\bibitem[McKay et~al.(1981)]{mckay1981practical}
Brendan~D McKay et~al.
\newblock Practical graph isomorphism.
\newblock \emph{Congressus Numerantium}, 30:\penalty0 45--87, 1981.

\bibitem[Miller(1980)]{miller1980isomorphism}
Gary Miller.
\newblock Isomorphism testing for graphs of bounded genus.
\newblock In \emph{Proceedings of the twelfth annual ACM symposium on Theory of
  computing}, pages 225--235, 1980.

\bibitem[Mowshowitz(1968)]{Mowshowitz1968}
Abbe Mowshowitz.
\newblock Entropy and the complexity of graphs: {I.} an index of the relative
  complexity of a graph.
\newblock \emph{The Bulletin of Mathematical Biophysics}, 30\penalty0
  (1):\penalty0 175--204, March 1968.

\bibitem[Mowshowitz and Dehmer(2010)]{mowshowitz2010symmetry}
Abbe Mowshowitz and Matthias Dehmer.
\newblock A symmetry index for graphs.
\newblock \emph{J. Math. Biophys}, 30:\penalty0 533--546, 2010.

\bibitem[Ostrowski et~al.(2011)Ostrowski, Linderoth, Rossi, and
  Smriglio]{ostrowski2011orbital}
James Ostrowski, Jeff Linderoth, Fabrizio Rossi, and Stefano Smriglio.
\newblock Orbital branching.
\newblock \emph{Mathematical Programming}, 126\penalty0 (1):\penalty0 147--178,
  2011.

\bibitem[Papadimitriou and Yannakakis(1991)]{papadimitriou1991optimization}
Christos~H. Papadimitriou and Mihalis Yannakakis.
\newblock Optimization, approximation, and complexity classes.
\newblock \emph{Journal of Computer and System Sciences}, 43\penalty0
  (3):\penalty0 425--440, December 1991.

\bibitem[Pedregosa et~al.(2011)Pedregosa, Varoquaux, Gramfort, Michel, Thirion,
  Grisel, Blondel, Prettenhofer, Weiss, Dubourg, Vanderplas, Passos,
  Cournapeau, Brucher, Perrot, and Duchesnay]{scikit-learn}
F.~Pedregosa, G.~Varoquaux, A.~Gramfort, V.~Michel, B.~Thirion, O.~Grisel,
  M.~Blondel, P.~Prettenhofer, R.~Weiss, V.~Dubourg, J.~Vanderplas, A.~Passos,
  D.~Cournapeau, M.~Brucher, M.~Perrot, and E.~Duchesnay.
\newblock Scikit-learn: Machine learning in {P}ython.
\newblock \emph{Journal of Machine Learning Research}, 12:\penalty0 2825--2830,
  2011.

\bibitem[Perez-Ortiz et~al.(2014)Perez-Ortiz, Gutierrez, and
  Hervas-Martinez]{perez2013projection}
Maria Perez-Ortiz, Pedro~Antonio Gutierrez, and Cesar Hervas-Martinez.
\newblock Projection-based ensemble learning for ordinal regression.
\newblock \emph{{IEEE} Transactions on Cybernetics}, 44\penalty0 (5):\penalty0
  681--694, May 2014.

\bibitem[Piperno(2008)]{piperno2008search}
Adolfo Piperno.
\newblock Search space contraction in canonical labeling of graphs.
\newblock \emph{arXiv preprint arXiv:0804.4881}, 2008.

\bibitem[Powell(1994)]{powell1994direct}
M.~J.~D. Powell.
\newblock A direct search optimization method that models the objective and
  constraint functions by linear interpolation.
\newblock In \emph{Advances in Optimization and Numerical Analysis}, pages
  51--67. Springer Netherlands, 1994.

\bibitem[Powell(1998)]{powell1998direct}
M.~J.~D. Powell.
\newblock Direct search algorithms for optimization calculations.
\newblock \emph{Acta Numerica}, 7:\penalty0 287--336, January 1998.

\bibitem[Qiu et~al.(2021)Qiu, Shaydulin, Liu, Alexeev, Henry, and
  Safro]{qiu2019elruna}
Zirou Qiu, Ruslan Shaydulin, Xiaoyuan Liu, Yuri Alexeev, Christopher~S Henry,
  and Ilya Safro.
\newblock Elruna: Elimination rule-based network alignment.
\newblock \emph{Journal of Experimental Algorithmics (JEA)}, 26:\penalty0
  1--32, 2021.

\bibitem[Rotman(2015)]{rotman2015advanced}
Joseph~J Rotman.
\newblock \emph{Advanced modern algebra}, volume 165.
\newblock American Mathematical Soc., 2015.

\bibitem[Sadrfaridpour et~al.(2019)Sadrfaridpour, Razzaghi, and
  Safro]{sadrfaridpour2019engineering}
Ehsan Sadrfaridpour, Talayeh Razzaghi, and Ilya Safro.
\newblock Engineering fast multilevel support vector machines.
\newblock \emph{Machine Learning}, 108\penalty0 (11):\penalty0 1879--1917, May
  2019.

\bibitem[Shaydulin and Alexeev(2019)]{Shaydulin2019EvaluatingDOI}
Ruslan Shaydulin and Yuri Alexeev.
\newblock Evaluating quantum approximate optimization algorithm: A case study.
\newblock In \emph{2019 Tenth International Green and Sustainable Computing
  Conference ({IGSC})}. {IEEE}, October 2019.

\bibitem[Shaydulin et~al.(2019)Shaydulin, Safro, and
  Larson]{Shaydulin2019MultistartDOI}
Ruslan Shaydulin, Ilya Safro, and Jeffrey Larson.
\newblock Multistart methods for quantum approximate optimization.
\newblock In \emph{2019 {IEEE} High Performance Extreme Computing Conference
  ({HPEC})}. {IEEE}, September 2019.

\bibitem[Simonyi(1995)]{simonyi1995graph}
G{\'a}bor Simonyi.
\newblock Graph entropy: a survey.
\newblock \emph{Combinatorial Optimization}, 20:\penalty0 399--441, 1995.

\bibitem[Szegedy(2019)]{Szegedy2019qaoaenergies}
Mario Szegedy.
\newblock What do {QAOA} energies reveal about graphs?
\newblock \emph{arXiv:1912.12277}, 2019.

\bibitem[Tange(2018)]{tange_ole_2018_1146014}
Ole Tange.
\newblock Gnu parallel 2018, 2018.
\newblock URL \url{https://doi.org/10.5281/zenodo.1146014}.
\newblock [Online].

\bibitem[Vazirani(2013)]{vazirani2013approximation}
Vijay~V Vazirani.
\newblock \emph{Approximation algorithms}.
\newblock Springer Science \& Business Media, 2013.

\bibitem[Verdon et~al.(2019)Verdon, Broughton, McClean, Sung, Babbush, Jiang,
  Neven, and Mohseni]{verdon2019learning}
Guillaume Verdon, Michael Broughton, Jarrod~R McClean, Kevin~J Sung, Ryan
  Babbush, Zhang Jiang, Hartmut Neven, and Masoud Mohseni.
\newblock Learning to learn with quantum neural networks via classical neural
  networks.
\newblock \emph{arXiv:1907.05415}, 2019.

\bibitem[Waegeman and Boullart(2009)]{waegeman2009ensemble}
Willem Waegeman and Luc Boullart.
\newblock An ensemble of weighted support vector machines for ordinal
  regression.
\newblock \emph{International Journal of Computer Systems Science and
  Engineering}, 3\penalty0 (1):\penalty0 47--51, 2009.

\bibitem[Wang et~al.(2018)Wang, Hadfield, Jiang, and Rieffel]{wang2018quantum}
Zhihui Wang, Stuart Hadfield, Zhang Jiang, and Eleanor~G. Rieffel.
\newblock Quantum approximate optimization algorithm for {MaxCut}: A fermionic
  view.
\newblock \emph{Physical Review A}, 97\penalty0 (2), February 2018.

\bibitem[Wilson et~al.(2019)Wilson, Stromswold, Wudarski, Hadfield, Tubman, and
  Rieffel]{wilson2019optimizing}
Max Wilson, Sam Stromswold, Filip Wudarski, Stuart Hadfield, Norm~M Tubman, and
  Eleanor Rieffel.
\newblock Optimizing quantum heuristics with meta-learning.
\newblock \emph{arXiv:1908.03185}, 2019.

\bibitem[Wurtz and Love(2020)]{Wurtz2020bounds}
Jonathan Wurtz and Peter~J. Love.
\newblock Bounds on {M}ax{C}ut {QAOA} performance for p{\textgreater}1.
\newblock \emph{arXiv:2010.11209}, 2020.

\bibitem[Zhou et~al.(2018)Zhou, Wang, Choi, Pichler, and
  Lukin]{zhou2018qaoaperformance}
Leo Zhou, Sheng-Tao Wang, Soonwon Choi, Hannes Pichler, and Mikhail~D. Lukin.
\newblock Quantum approximate optimization algorithm: Performance, mechanism,
  and implementation on near-term devices.
\newblock \emph{arXiv:1812.01041}, 2018.

\end{thebibliography}

\vfill
\framebox{\parbox{\textwidth}{The submitted manuscript has been created by
UChicago Argonne, LLC, Operator of Argonne National Laboratory (``Argonne'').
Argonne, a U.S.\ Department of Energy Office of Science laboratory, is operated
under Contract No.\ DE-AC02-06CH11357.  The U.S.\ Government retains for itself,
and others acting on its behalf, a paid-up nonexclusive, irrevocable worldwide
license in said article to reproduce, prepare derivative works, distribute
copies to the public, and perform publicly and display publicly, by or on
behalf of the Government.  The Department of Energy will provide public access
to these results of federally sponsored research in accordance with the DOE
Public Access Plan \url{http://energy.gov/downloads/doe-public-access-plan}.}}

\end{document}